\documentclass[copyright,creativecommons]{eptcs}
\providecommand{\event}{QPL 2025} 
\usepackage{iftex}
\input{qutrit.sty}
\ifpdf
  \usepackage{underscore}         
  \usepackage[T1]{fontenc}        
\else
  \usepackage{breakurl}           
\fi

\title{A Complete and Natural Rule Set \\ for Multi-Qutrit Clifford Circuits}

\author{Sarah Meng Li\footnote{Informatics Institute, University of Amsterdam, Amsterdam, Netherlands.}\; \footnote{Department of Combinatorics \& Optimization, University of Waterloo, Waterloo, Canada.}\; \footnote{Institute for Quantum Computing, University of Waterloo, Waterloo, Canada.}\; \footnote{Perimeter Institute for Theoretical Physics, Waterloo, Canada.}\; \email{sarah.li@uva.nl}\and  Michele Mosca \footnotemark[2]\; \footnotemark[3] \;\footnotemark[4] \email{michele.mosca@uwaterloo.ca}\and  Neil J. Ross \footnote{Department of Mathematics and Statistics, Dalhousie University, Halifax, Canada.} \email{neil.jr.ross@dal.ca}\and John van de Wetering \footnotemark[1] \email{john@vdwetering.name} \and Yuming Zhao \footnote{Department of Pure Mathematics, University of Waterloo, Waterloo, Canada.}\; \footnotemark[3] \;\footnote{QMATH, Department of Mathematical Sciences, University of Copenhagen, Copenhagen, Denmark.} \email{yuming@math.ku.dk}
}
\def\titlerunning{A Complete and Natural Rule Set for Multi-Qutrit Clifford Circuits}
\def\authorrunning{Li, S. M., Mosca, M., Ross, N. J., van de Wetering, J., Zhao, Y.}

\begin{document}
\maketitle

\begin{abstract}
We present a complete set of rewrite rules for $n$-qutrit Clifford circuits where $n$ is any non-negative integer. This is the first completeness result for any fragment of quantum circuits in odd prime dimensions. We first generalize Selinger’s normal form for $n$-qubit Clifford circuits to the qutrit setting. Then, we present a rewrite system by which any Clifford circuit can be reduced to this normal form. We then simplify the rewrite rules in this procedure to a small natural set of rules, giving a clean presentation of the group of qutrit Clifford unitaries in terms of generators and relations. 
\end{abstract}

\section{Introduction}
\label{sec:introduction}
Most research on quantum computing has focused on qu\emph{bits}, two-dimensional quantum systems. In recent years, research on qu\emph{dit} quantum computing, where the fundamental state space has a higher dimension than two, has witnessed a surge in interest and activity~\cite{chi2022programmable,goss2022high,Got99,gustafson2025synthesis,nikolaeva2024scalable,Rin21,PhysRevLett.134.050601,wang2020qudits}. Qudits offer richer algebraic structures that result in unique error correction capabilities~\cite{PhysRevLett.113.230501,FG24}, lower overhead magic state distillation~\cite{PhysRevX.2.041021,PS24}, stronger quantum correlations~\cite{Bru08,Des22}, and improvements to various quantum algorithms~\cite{PhysRevA.96.012306,PhysRevA.103.032417}. On the physical side, most realizations of qubits naturally have higher energy states so they could readily encode qudits. This allows for greater information density, but comes at the cost of increased difficulty in system control~\cite{Blo21,Hrm22,Low20,Rin21,Win23,Ye18,Yur20}.

In this paper, we focus on qu\emph{trits}, the three-dimensional quantum systems. Qutrits strike a balance between enabling the benefits of working with higher-dimensional systems, while not being too complicated to prevent hardware implementations. In particular, we study the group of qutrit \emph{Clifford unitaries}, the unitaries that normalize the qutrit analogue of the Pauli group. As for qubits, the qutrit Clifford group plays a crucial role in the stabilizer formalism for ternary quantum systems, serving as the backbone for designing fault-tolerant quantum gates and protocols~\cite{aigner2025qudit,Hus12,PhysRevLett.115.030501,Got99,PhysRevA.102.042409}.

Using a small set of generators for the qutrit Clifford group---the Hadamard, S and CZ gates---we find a simple set of relations that suffice to prove all equalities between Clifford circuits. That is, we prove that the rule set in \cref{fig:rewriterules} is \emph{complete}. This is the first completeness result for a non-trivial quantum circuit fragment for higher-dimensional qudits. Our proof strategy is generic and is intended to generalize to qudits of higher dimensions.

\begin{figure}[!thb]
    \begin{equation*}
        \scalebox{0.7}{\input{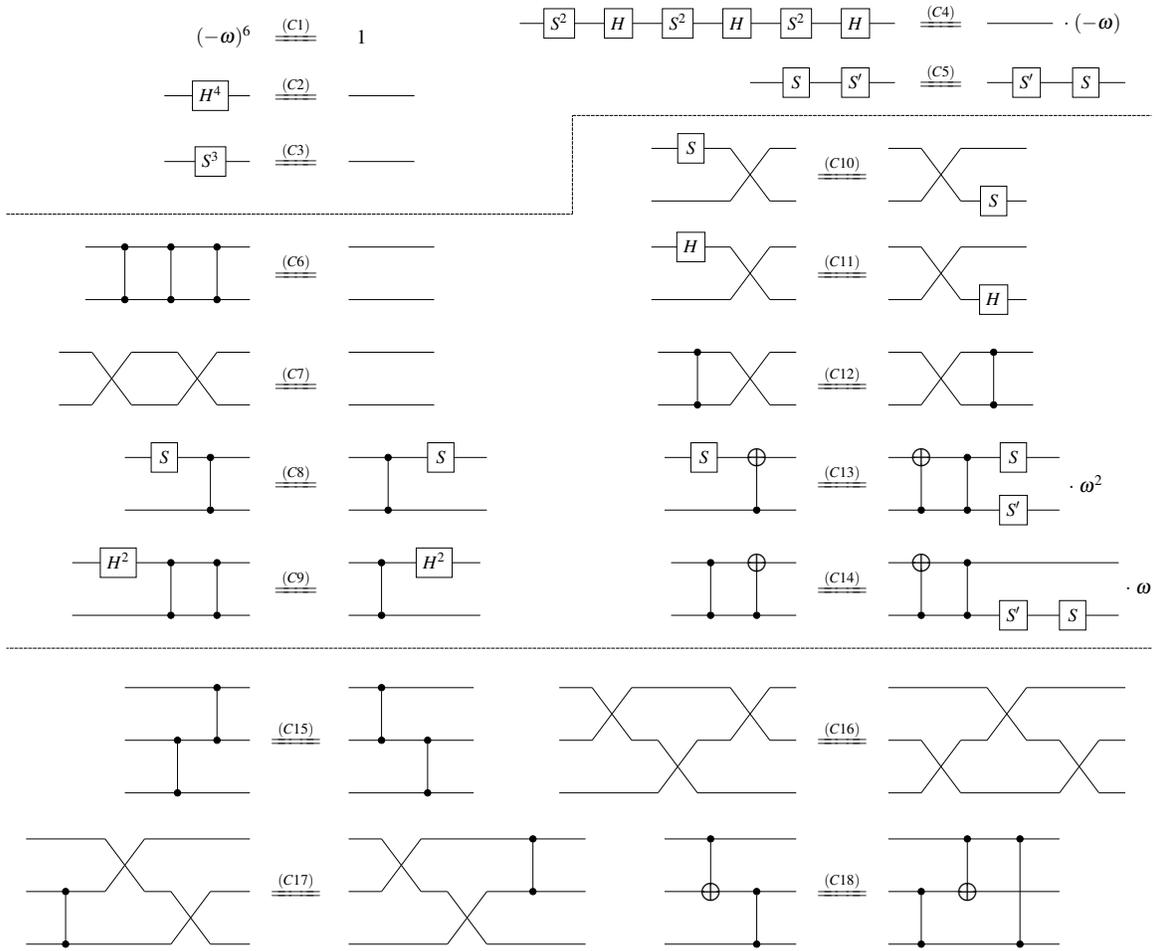}}
    \end{equation*}
    \vspace{-.4cm}
  \caption{A complete set of rewrite rules for $n$-qutrit Clifford circuits. Note that $\omega$, $S'$, $\SWAP$, $\CX$, $\XC$, and the $\CZ$ acting on the first and third qutrits are derived generators defined in \cref{fig:summary-derived-generators}. Let $g$ be a gate and $m$ be a natural number. $g^m$ denotes $m$ copies of $g$ gate. In addition to these circuit relations, we assume the standard spatial relations which allow us to commute gates acting on disjoint subsets of qutrits. In the remainder of this paper, we use `relations' and `(rewrite) rules' interchangeably.}
  \label{fig:rewriterules}
\end{figure}

A complete graphical calculus for qutrit linear maps in the stabilizer fragment exists using the ZX-calculus~\cite{poor2023qupit,wang2018qutrit}. However, ZX-diagrams represent general linear maps, not just unitaries. Hence, this completeness result does not restrict to the unitary Clifford group in any straightforward way. In general, a completeness result through a graphical calculus for all linear maps of a certain type does not seem to easily imply a method of deriving a calculus for just the unitaries. Consider for instance that a complete ZX-calculus for universal qubit linear maps was found in 2017~\cite{jeandel2018diagrammatic,ng2017universal}, but it was not until 2022 that a complete calculus of universal quantum circuits appeared and it used very different methods~\cite{clement2023complete}. Moreover, there is no complete calculus for, for instance, Toffoli-Hadamard circuits, even though there is one for the corresponding set of linear maps~\cite{backens2023completeness}.

One way to prove completeness for a fragment of quantum circuits (i.e., a collection of unitary operators) is to define a unique normal form for each unitary and to show that every circuit can be rewritten to this normal form. This is how the completeness result for qubit Clifford circuits was established in~\cite{makary2021generators,selinger2015generators}. Our proof follows along similar lines, but with significant complications due to the increased system dimension and the additional degrees of freedom. At a high level, we can view the normal form proposed by~\cite{selinger2015generators} as encoding the \emph{stabilizer tableau} of a Clifford unitary. Each part of the normal form determines the image of Pauli generators under the action of this Clifford unitary. To prove completeness, we need to show that when a Clifford circuit is in normal form and it is composed with a Clifford generator (i.e., a Hadamard, S, or CZ gate), the resulting circuit can be rewritten into a new normal form. This process, that we call \emph{Clifford normalization}, is sketched in \cref{fig:normalization2}. Then our problem is reduced to finding rewrite rules to push Clifford generators through each part of the normal form in a systematic manner.

\vspace{-.5 cm}

\begin{figure}[!htb]
    \centering
    \[
 \scalebox{.85}{\input{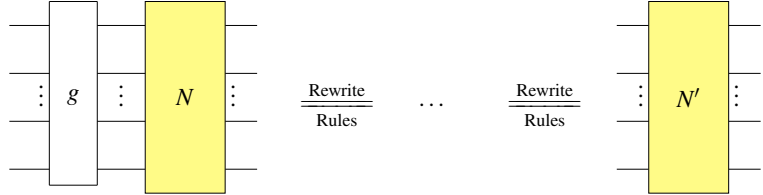}}
\]
\vspace{-.3 cm}
    \caption{$N$ is a Clifford circuit in normal form. $g$ is either an $H$ or $S$ gate acting on some qutrit, or a $\CZ$ gate acting on two adjacent qutrits. After pushing $g$ through each part of $N$, the normal form is updated to $N'$.}
    \label{fig:normalization2}
\end{figure}

The problem in generalizing this to qutrits is that the choice of each normal form component, what we call a \emph{normal box}, is not unique, but has some degrees of freedom. Different choices result in different `dirty gates' appearing when we push a generator through a normal box. \cref{fig:residual-dirty-gates-2} shows a concrete example. Making a wrong choice might produce dirty gates that cannot be properly pushed through the subsequent normal boxes. In the worst case, the Clifford normalization will not terminate. Another complication is that there are many more combinations of normal boxes and generators to check due to the increased dimension. All of this means that in the qubit case, the design space was small enough for a brute-force approach to succeed in finding a suitable way to construct the normal boxes, but with qutrits, the significantly larger design space makes such an approach much more challenging.

\begin{figure}[!htb]
\[
\scalebox{1}{\input{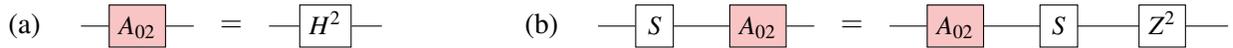}}
\]

\vspace{-.3 cm}

\caption{(a) shows a single-qutrit normal box called $A_{02}$, which in this particular case is equal to $H^2$. (b) shows a \emph{box relation}, where an $S$ gate is pushed through $A_{02}$ and produces dirty gates $S$ and $Z^2$. $H^2$ and $Z^2$ denote two copies of Hadamard and Pauli $Z$ gates respectively.}
\label{fig:residual-dirty-gates-2}
\end{figure}

Instead, we identify additional symmetries that each normal box must satisfy, which then guarantee that dirty gates have the desired form for most of the box relations. This restricts the design space and allows us to verify on a higher level that our normal boxes have the right properties for the Clifford normalization to terminate (see \cref{sec:design} for the details). With the correct definitions of the normal boxes in hand, we can then write down all the relations needed to push the generators through the normal boxes, what we call the \emph{box relations}. In \cref{app:relations}, we present the $380$ box relations, which together are complete for multi-qutrit Clifford circuits. 

The second major component of this work is to reduce the $380$ box relations to a much smaller, more canonical set of rewrite rules. In the supplement~\cite{Supplement2024}, we prove that the $18$ gate relations presented in~\cref{fig:rewriterules} imply these $380$ box relations. In order to present these rules in the nicest possible fashion, we use some derived generators that are defined in~\cref{fig:summary-derived-generators}.

Overall, \cref{fig:rewriterules} consists of natural and intuitive rules regarding the order of generators, commutations of generators, and properties of the $\SWAP$ gate. The only exception to this is (C4), which can be understood as describing an Euler decomposition of the Hadamard gate. Up to a global phase, it can be reframed with some rewriting to $H=S^2(HS^2H^{-1})(S')^2$, meaning $H$ can be decomposed into a $Z$-, $X$- and then another $Z$-rotation. Note that the derived generator $S'= H^2SH^2$ in \cref{fig:summary-derived-generators} has no qubit counterpart, but it is a $Z$ rotation that we can use together with $S$ to define the Pauli $Z$ gate, $Z=(S')^2 S$. See also the discussion in Section 1.1 of \cite{Supplement2024}. The commutation of $S$ and $S'$ in (C5) is the only single-qutrit rule that has no qubit counterpart.

Generalizing from qutrits to qu\emph{pits}, where each wire carries a $p$-dimensional quantum system for an odd prime $p$, we can construct a version of \cref{fig:rewriterules} with sound equations for qupit Clifford operators. Although proving that such a generalized rule set is complete still requires substantial work, \cref{fig:rewriterules} presents a natural framework for understanding gate interactions. This, in turn, offers a promising path toward extending the qutrit Clifford completeness to other odd prime dimensions.

The rest of the paper is organized as follows. In \cref{sec:preliminaries}, we introduce the basic concepts and conventions that will be used throughout this paper. In \cref{sec:normal-form}, we define normal forms and prove that every multi-qutrit Clifford circuit admits a unique normal form. In \cref{sec:normalization-with-relations}, we describe how to rewrite any Clifford circuit to this normal form, based on which we find a complete set of rewrite rules for qutrit Clifford circuits. We end with some final remarks and open questions in \cref{sec:conclusion}.
\section{Preliminaries}
\label{sec:preliminaries}
We are interested in \emph{matrices}, which we sometimes call \emph{operators}, or \emph{gates}. We write quantum circuits from left to right, which is in the opposite order of writing the matrix multiplication.

Let $m\in \N$. We write $I_m$ for the $m\times m$ identity matrix, dropping the subscript $m$ when the dimension is clear from the context. An $n$-qutrit operator is a $3^n \times 3^n$ matrix. We say that an $n$-qutrit operator $U$ is a \emph{scalar}, if it is a scalar multiple of the identity operator, i.e., $U = \lambda I$, for some $\lambda \in \C$. For simplicity, we write $U = \lambda$ in such a case. In what follows, we use `scalar' and `phase' interchangeably.

\subsection{Clifford Groups and Circuits}
\label{subsec:unitaries}

We fix $\omega = e^{2\pi i/3}$ and $-\omega = e^{-2\pi i/6}$ as our \emph{primitive third} and \emph{sixth roots of unity} respectively. 
\begin{definition}
    The single-qutrit \emph{Pauli $X$ and $Z$ gates} are defined as follows
    \[
            X=\begin{bmatrix}
            0 & 0 & 1\\
            1 & 0 & 0\\
            0 & 1 & 0
        \end{bmatrix}
        \qquad
        \mbox{and}
        \qquad
        Z=\begin{bmatrix}
            1 & 0 & 0\\
            0 & \omega & 0\\
            0 & 0 & \omega^2
        \end{bmatrix}.
    \]        
\end{definition}

For any $m \in \N$, we write $\Z_m =\{0, 1, \ldots, m-1\}$ equipped with the standard addition modulo $m$. 
\begin{definition}
    \label{def:Pauli-groups}
    For $n \in \N$, the \emph{$n$-qutrit Pauli group} $\Pauli_n$ is defined as
    \[
    \Pauli_n = \{ \omega^c P_1 \otimes \cdots \otimes P_{n};\; c\in \Z_3 \mbox{ and } P_j = X^{a_j}Z^{b_j} \mbox{ for some } a_j,b_j \in\Z_3\}.
    \]
\end{definition}

We now define the \emph{Clifford group}, which will be the focus of this paper. As \cref{def:Clifford-unitaries} states, the $n$-qutrit Clifford group is the normalizer of the Pauli group in the $3^n\times 3^n$ unitary group $\mathcal{U}(3^n)$.

\begin{definition}
\label{def:Clifford-unitaries}
    For $n \in \N$, the \emph{$n$-qutrit Clifford group} $\Clifford_n$ is defined as 
    \[
    \Clifford_n = \{U \in \mathcal{U}({3^n});\; UPU^\dagger \in \Pauli_n  \mbox{ for all } P \in \Pauli_n\}.
    \]
\end{definition}

\begin{definition}
The single-qutrit \emph{$H$ and $S$ gates} and the two-qutrit \emph{$\CZ$ gate} are defined as follows
\[
        H=\frac{1}{\omega^2 - \omega}\begin{bmatrix}
            1 & 1 & 1\\
            1 & \omega & \omega^2\\
            1 & \omega^2 & \omega
        \end{bmatrix}, 
        \quad 
        S=\begin{bmatrix}
            \omega & 0 & 0\\
            0 & \omega & 0\\
            0 & 0 & 1
        \end{bmatrix}, 
        \quad
        \mbox{and}
        \quad
        \CZ=\diag(1, 1, 1, 1, \omega, \omega^2, 1, \omega^2, \omega).
\]
\label{def:H-S-gates}
\end{definition}

It can be verified by direct calculation that $H,S\in\Clifford_1$ and $\CZ\in\Clifford_2$. The $H$ gate is also known as the \emph{Hadamard gate}, while the $\CZ$ gate is also known as the \emph{controlled-Z gate}. We chose the global phase in the $H$ and $S$ gates to make our results simpler to state. 

Note that \cref{def:Clifford-unitaries} considers all possible global phases $e^{i\alpha}$ as being Clifford, where $\alpha\in \R$. This makes the group uncountably infinite, but $H$, $S$ and $\CZ$ as matrices have entries in the ring $\Z[1/3,\omega]$ and hence can only represent a subset of these phases. In what follows, we restrict our attention to the subset of Clifford operators that can be represented as matrices over $\Z[1/3,\omega]$. As a result, \cref{subsec:scalars} shows that the global phases can only be powers of $-\omega$. In \cref{prop:-normal-form}, we show that $-\omega$, $H$, $S$, and $\CZ$ generate $\Clifford_n$.

It is clear from \cref{def:Clifford-unitaries}, that $\Pauli_n \subseteq \Clifford_n$. In fact, $Z$ and $X$ can be expressed using $H$ and $S$ as $Z=SH^2S^2H^2$ and $X=HZH^3$. Note that when we write in text form, a juxtaposition of gates means we do the usual matrix multiplication. Since the matrices of $Z$, $S$, and $\CZ$ gates are diagonal, we sometimes refer to them as \emph{diagonal gates}.

We can also denote Clifford operators using the diagrammatic language of quantum circuits. \cref{fig:circuit-notation} shows the circuit representation of the $H$, $S$, and $\CZ$ gates (as well as that of the identity). 

\vspace{-.3 cm}

\begin{figure}[!htb]
        \centering
        \[
        \scalebox{1}{\input{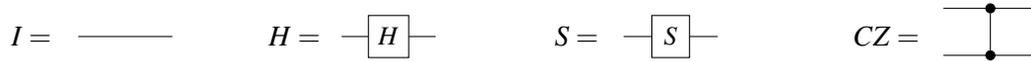}}
        \]
        \vspace{-.3 cm}
        \caption{Circuit notation for the one- and two-qutrit Clifford generators.}
        \label{fig:circuit-notation}
    \end{figure}

Next, we introduce \emph{derived} generators that are summarized in~\cref{fig:summary-derived-generators}. They are useful for defining the Clifford normal forms in \cref{sec:normal-form}, designing the normal boxes in \cref{sec:design}, deriving box relations in \cref{app:relations}, and carrying out the relation reduction in the supplement~\cite{Supplement2024}. One can think of them as circuit shorthands over the generators $-\omega$, $H$, $S$, and $\CZ$ gates. $-1$ (D1) and $\omega$ (D2) are used to simplify the expression of box relations. As opposed to the qubit $H$ gate, the qutrit $H$ gate is of order 4, $H^4 = 1$. Up to the global phase $-1$, $H^2$ acts as a `negation', $H^2\ket{x} = \ket{-x}$. Here, the negation is taken modulo 3, so that $\ket{0}\mapsto \ket{0}$, $\ket{1}\mapsto \ket{2}$, and $\ket{2}\mapsto \ket{1}$. We then see that $S'$ (D3) is also a diagonal gate, and as a matrix $S' = \diag(\omega, 1, \omega)$.

\CX (D7) is short for a \emph{controlled-NOT gate}. The control and target of a \CX gate are on the top and bottom qutrit wires, respectively. Let $x,y\in \Z_3$. $\CX\ket{x,y} = \ket{x,x+y}$, where the addition is taken in $\Z_3$ and hence modulo 3. A \XC gate (D8) has an opposite control-target arrangement. 

We assume that the multi-qutrit (derived) generators only act on adjacent qutrits and we call them \emph{local gates}. This is without loss of generality, as gates on non-adjacent qutrits can be equivalently expressed using SWAP gates. In \cref{fig:summary-derived-generators}, we write a gate $G$ acting on two non-adjacent qutrits as a circuit shorthand for a local gate surrounded by \SWAP gates. (D9), (D10), and (D11) are called the \emph{remote $\CZ$, $\CX$, and $\XC$ gates}.

\begin{figure}[!htb]
    \[
    \scalebox{.75}{\input{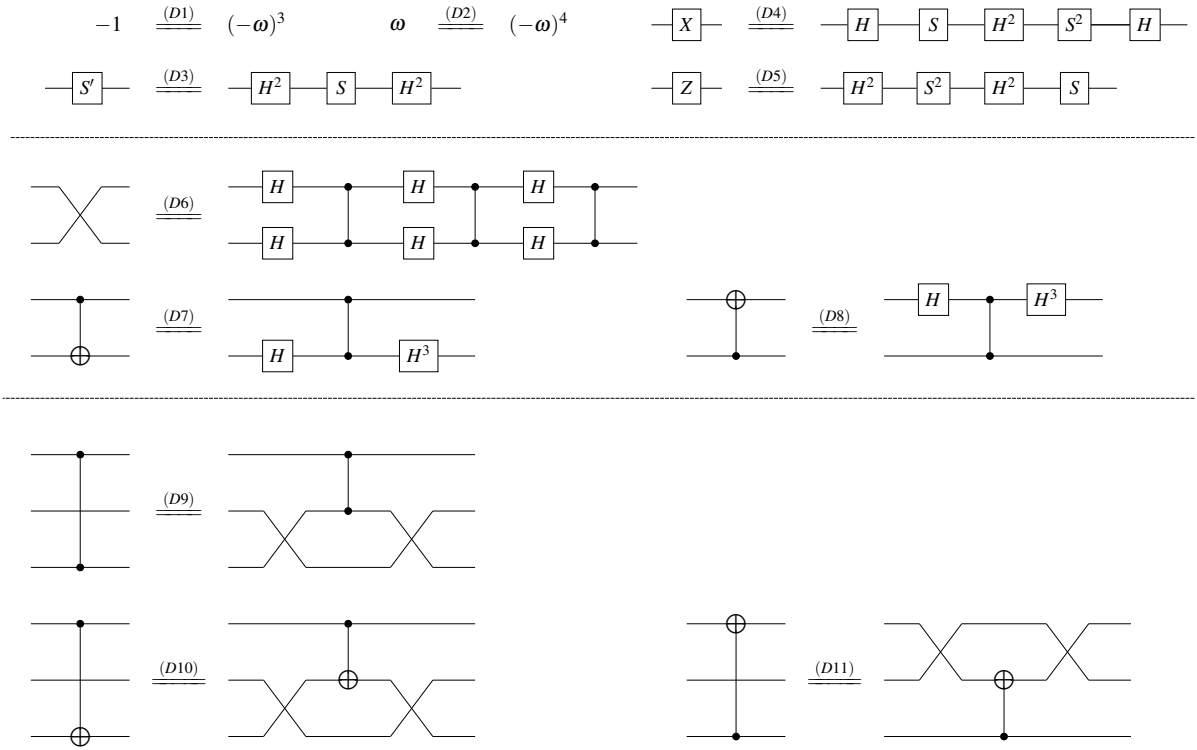}}
    \]
    \vspace{-.3 cm}
    \caption{Circuit notation for the derived Clifford generators.}
    \label{fig:summary-derived-generators}
\end{figure}

Finally, we introduce the notation conventions used in the rest of this paper. For a circuit $C \in \Clifford_n$, its global phase $\lambda\in\C$ is placed at the end of the circuit, as shown in \cref{fig:circuit-conventions}.(a). Let $m \in \N$. In \cref{fig:circuit-conventions}.(b), we write $C^m$ to denote the $m$-fold composition of $C$ with itself.

\begin{figure}[!htb]
    \centering
    \input{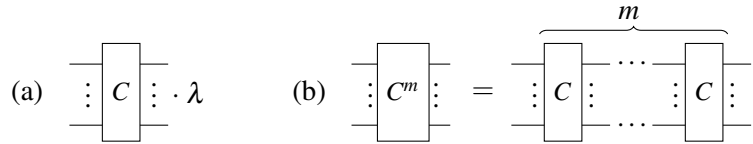}
    \vspace{-.1 cm}
    \caption{Notation conventions.}
    \label{fig:circuit-conventions}
\end{figure}

\vspace{-.4 cm}

\begin{remark}\label{rem:groupoid}
    Instead of viewing each $\Clifford_n$ as a separate group, we could consider the collection of all these groups as a single \emph{groupoid}: a category where every morphism is an isomorphism. In this category, the objects are $\N$, corresponding to the number of qutrits. This then has a monoidal structure given by $a\otimes b := a+b$ on objects $a,b\in\N$. In this groupoid, we get the standard coherence isomorphism $(f\otimes \text{id})\circ (\text{id}\otimes g) = (\text{id}\otimes g) \circ (f\otimes \text{id})$, allowing us to commute arbitrary gates acting on disjoint sets of qutrits. Throughout the paper, we will assume this property without further comment, as was done in~\cite{selinger2015generators}.
\end{remark}

\subsection{Clifford Conjugation on the Pauli Group}
\label{subsec:action}

By definition, a Clifford operator $C\in\Clifford_n$ maps any Pauli operator $P\in \Pauli_n$ to another Pauli operator, $CPC^\dagger \in \Pauli_n$. This conjugation operation is denoted as $C \bullet P = CPC^\dagger$, and we refer to it as \emph{the Clifford $C$'s action on the Pauli operator $P$}. Let $Q = C\bullet P$. Then $CP = QC$, and we can interpret this as `pushing $P$ through $C$ gives us $Q$'. It will be useful to introduce a circuit shorthand for this.

\begin{definition}\label{def:conjugation}
    Let $C \in \Clifford_n$ and $P,Q \in \Pauli_n$ with $Q = C\bullet P$. Write $P = P_1 \otimes \ldots \otimes P_{n}$ and $Q = Q_1 \otimes \ldots \otimes Q_{n}$ for $P_j, Q_j \in \Pauli_1$. Then the conjugation relation $Q = C\bullet P$ is denoted as
    \[
    \scalebox{.9}{\input{figures/Preliminaries/CliffordConjugation.tikz}}
    \]
\end{definition}

Following this convention, \cref{app:generator-actions} summarizes the actions of $H$, $S$, $X$, $Z$, $\CZ$, and $\SWAP$ gates on Pauli operators.

\begin{lemma}
    Let $C \in \Clifford_n$, the action of $C$ on $\Pauli_n$ is a group automorphism of $\Pauli_n$ and it fixes scalars.
    \label{lem:group-automorphism}
\end{lemma}
\begin{proof}
It is sufficient to show that for all $P, Q \in \Pauli_n$, $C \bullet (PQ) = C(PQ) C^\dagger = (CPC^\dagger) (CQC^\dagger) = (C \bullet P)(C \bullet Q)$. For all $\lambda \in \C$, since $C\bullet \lambda = C\lambda C^\dagger = \lambda C C^\dagger = \lambda$, this automorphism fixes scalars.
\end{proof}

\begin{proposition}\label{prop:scalars}
    Let $C \in \Clifford_n$. If $C \bullet P = P$ for all $P \in \Pauli_n$, then C is a scalar (i.e.~proportional to the identity).
\end{proposition}
\begin{proof}
    Any complex $3 \times 3$ matrix can be written in the form $\sum_{j=1}^9c_jP_j, \; c_j \in \C, \; P_j \in \{X^aZ^b; \; a, b \in \Z_3\}$. It follows that $\Pauli_n$ spans the vector space formed by $3^n\times 3^n$ complex matrices. Since $C \bullet P = P$ for all $P \in \Pauli_n$, $C \bullet M = CMC^\dagger = M$ for all operators $M$. Hence $CM = MC$ means $C$ must commute with all matrices. Then $C$ must be a scalar.
\end{proof}

\begin{corollary}
    Let $C_1,C_2 \in \Clifford_n$. Suppose for all $P\in\Pauli_n$, $C_1\bullet P = C_2\bullet P$. Then $C_1 = \lambda C_2$, $\lambda \in \C$.
    \label{cor:Clifford-stabilisers}
\end{corollary}

\begin{proof}
    For all $P\in\Pauli_n$, $C_1 P C_1^\dagger = C_2 P C_2^\dagger$. This implies that $C_2^\dagger C_1 P C_1^\dagger C_2 = P$. Take $U = C_2^\dagger C_1$, then $U^\dagger = (C_2^\dagger C_1)^\dagger = C_1^\dagger C_2$. Hence $UPU^\dagger = P$ for all $P\in \Pauli_n$. By \cref{prop:scalars}, $U =\lambda$, for some $\lambda \in \C$. It follows that $C_1 = \lambda C_2$. This completes the proof.
\end{proof}

It is well-known that a Clifford unitary $C$ is fully determined by its action on Pauli operators. Since the Clifford conjugation is a group homomorphism, it suffices to know $C$ acts on a set of \emph{Pauli generators}. In \cref{def:Pauli-basis}, we use $\mathcal{B}_{n}$ to denote a collection of Pauli operators that has a Pauli $Z$ or $X$ acting on a single qutrit, and identities everywhere else. $\mathcal{B}_{n}$ generates $\mathcal{P}_n$ and it is used in the remainder of this paper.

\begin{definition}
    Let $n\in \N$. $Z^{(j)}$ denotes the $n$-qutrit Pauli operator $P=P_1\otimes\cdots\otimes P_{n}$ with $P_\ell = Z$ if $\ell= j$, and $P_\ell = I$ otherwise. We write $X^{(j)}$ for the analogous $n$-qutrit Pauli operator by substituting $Z$ with $X$ in the above definition. $\mathcal{B}_{n}=\{X^{(j)},Z^{(j)};\;1\leq j\leq n\}$.
    \label{def:Pauli-basis}
\end{definition}


Thus, for $1\leq j \leq n$ and $P^{(j)},Q^{(j)}\in \Pauli_n$, knowing $P^{(j)} = C\bullet Z^{(j)}$ and $Q^{(j)} = C\bullet X^{(j)}$ fully determines $C$ (up to a global phase). The collection of these $2n$ maps is often called the \emph{stabilizer tableau} of $C$~\cite{PhysRevA.70.052328,KissingerWetering2024Book}. This notion will become useful for designing the multi-qutrit Clifford normal form. 

When $n=1$, the stabilizer tableau is a $3\times 2$ matrix over $\Z_3$. \cref{fig:tableau} shows how to encode a single-qutrit Clifford operator $C=SHSH$ using a tableau. Let $P,Q\in \Pauli_1$ be the images of $Z$ and $X$ in $C$ respectively. They are encoded by the first and second columns of the tableau. The first row encodes the $\omega$-exponent in the phase of $P$ and $Q$. The second and third rows denote the participation of $Z$ and $X$ in $P$ and $Q$, respectively. This notation can be generalized to encode an $n$-qutrit Clifford operator using a $(2n+1)\times 2n$ matrix over $\Z_3$. 

\begin{figure}[!htb]
    \centering
    \[
    \scalebox{.9}{\input{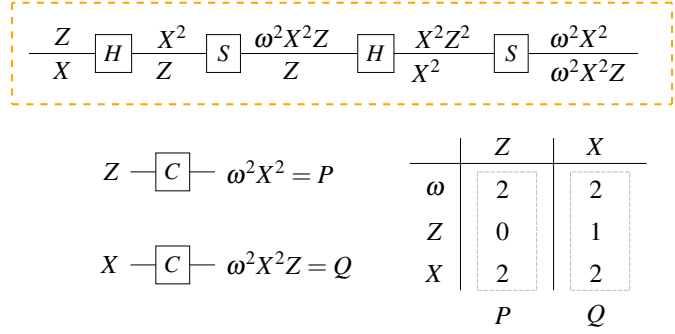}}
    \]

    \vspace{-.5 cm}
    
    \caption{According to \cref{app:generator-actions}, we can track the Pauli evolution in $C=SHSH$ and find the images of $Z$ and $X$. Since $Z$ is mapped to $\omega^2X^2$ and $X$ is mapped to $\omega^2X^2Z$, we can construct a $3 \times 2$ stabilizer tableau for $C$.}
    \label{fig:tableau}
\end{figure}

\begin{remark}
    By \cref{lem:group-automorphism}, Clifford conjugation on Paulis is a group automorphism. Hence, $P^{(j)}$ and $Q^{(j)}$ inherit the commuting and non-commuting relations from $Z^{(j)}$ and $X^{(j)}$. That is, $P^{(j)}Q^{(j)} = \omega Q^{(j)}P^{(j)}$, and $P^{(j)}Q^{(\ell)} = Q^{(\ell)} P^{(j)}$ when $j\neq \ell$.
    \label{rmk:automorphism}
\end{remark}
\section{A Normal Form for Multi-Qutrit Clifford Circuits}
\label{sec:normal-form}

Similar to the qubit Clifford normal form defined in~\cite{makary2021generators,selinger2015generators}, our normal form for a qutrit Clifford circuit $C\in\Clifford_n$ is based on uniquely representing its stabilizer tableau, and hence uniquely representing this operator. First, we provide intuitions for designing such a normal form. Then, we lay out the technical details of constructing this normal form.

\paragraph{Intuitions}

It will be helpful to think of this normal form as synthesizing the \emph{inverse} of $C$, and we will refer to this as the \emph{inverse synthesis}. Suppose $C^{-1}$ has a stabilizer tableau, for $1 \leq j \leq n$,

\begin{equation}
    P^{(j)} = C^{-1}\bullet Z^{(j)}, \qquad  Q^{(j)} = C^{-1}\bullet X^{(j)}.
    \label{eq:map-1}
\end{equation}

We want to find a Clifford circuit in normal form $N$ such that

\begin{equation}
    N\bullet P^{(j)} = Z^{(j)},\qquad N\bullet Q^{(j)} = X^{(j)}.
    \label{eq:map-2}
\end{equation}

Putting \eqref{eq:map-1} and \eqref{eq:map-2} together, \eqref{eq:map-3} shows that $(NC^{-1})\bullet P = P$ for all $P\in \Pauli_n$. By \cref{prop:scalars}, $NC^{-1}=\lambda$, for some $\lambda\in \C$. Hence, $N=\lambda C$. Then we find the normal form of $C$, up to some global phase. 

\begin{align}
    N\bullet P^{(j)} = NC^{-1} \bullet Z^{(j)}=Z^{(j)},\qquad N\bullet Q^{(j)} = NC^{-1} \bullet X^{(j)}=X^{(j)},\qquad 1\leq j \leq n.
    \label{eq:map-3}
\end{align}

In a nutshell, we carry out the inverse synthesis by induction. Firstly, from \eqref{eq:map-1}, we have

\[
    \scalebox{.8}{\input{figures/NormalForm/C-Cinverse.tikz}}
    \]

Next, find a circuit $T^{(n)}$ in normal form that sends $P^{(n)}$ to $Z^{(n)}$ and $Q^{(n)}$ to $X^{(n)}$. We then have

\[
    \scalebox{.85}{\input{figures/NormalForm/C-Cinverse2.tikz}}
    \]

This implies that $T^{(n)}C^{-1}$ acts as an identity on the last qutrit, so we can treat it as a Clifford operator acting on the first $n-1$ qutrits. For this `smaller' circuit, repeat the same procedure and find a circuit $T^{(n-1)}$ in normal form such that $T^{(n-1)}T^{(n)}C^{-1}$ acts as an identity on the last two qutrits. By iteratively concatenating $T^{(j)}$, from $j=n$ to $j=1$, we obtain $N = T^{(1)}\cdots T^{(n-1)}T^{(n)}$. 

\[
    \scalebox{.85}{\input{figures/NormalForm/C-Cinverse3.tikz}}
    \]

One can think of the inverse synthesis as simplifying the stabilizer tableau by performing Gaussian elimination on its bottom $2n \times 2n$ matrix. As a simple example, when $n=1$ and $C=(SHSH)^{-1}$, the inverse synthesis reduces the stabilizer tableau of $C^{-1}$ to an all-zero row followed by an identity matrix.

\[
    \scalebox{.8}{\input{figures/NormalForm/CliffordCircuit2.tikz}}
    \]

\paragraph{Technical Details}

Our goal is to map the Pauli generators to the desired outcomes by using some elementary building blocks, which we call \emph{normal boxes}. To obtain a unique normal form, we need to identify some notion of uniqueness when designing these boxes. In total, there are six \emph{types} of normal boxes. They are labelled as $A$, $B$, $C$, $D$, $E$, and $F$, as shown in \cref{fig:normal-box-auto}. Each type of normal box has different \emph{variants}, and each variant is labelled by the subscript of the box.

\begin{definition}
    For $K \in \{A, B, C, D, E, F\}$, we call a normal box of \emph{type} $K$ a \emph{$K$ normal box}. We use $K_k$ to denote a \emph{variant} of the $K$ normal box, where $k$ can be any allowed indices in the first column of \cref{fig:normal-box-auto}.
\end{definition}

In the second column of \cref{fig:normal-box-auto}, we specify the required action of each normal box on some Paulis. Take $B$ normal boxes as an example. The specification of other normal boxes can be understood analogously. For $a,b\in \Z_3$, $B_{ab}$ maps $X^aZ^b \otimes Z$ to $Z \otimes I$. The uniqueness of these required actions determines the uniqueness of each $B_{ab}$ box. However, this action does not specify a full stabilizer tableau, so it is not sufficient to determine which unitary we should use for each of these boxes. Different choices of implementation might result in different designs of the normal form. Making the `wrong' choice might make it more difficult to establish the multi-qutrit Clifford completeness. 

In the third column of \cref{fig:normal-box-auto}, we identify additional actions that some normal boxes should have on Paulis. This takes advantage of the underlying mechanics of our normal form and further restricts the design space. In \cref{sec:design}, we discuss in detail how these restrictions come into play.

 \begin{figure}[!tbh]
     \centering
      \[
\scalebox{.9}{\input{figures/normalBoxesAutomorphisms2.tikz}}
    \]
    \vspace{-0.5cm}
     \caption{There are six types of normal boxes. The colour of each box (red or green) specifies whether it belongs to the Z-part or the X-part of the normal form.}
     \label{fig:normal-box-auto}
 \end{figure}

\cref{fig:Z-and-X-Normal-Boxes} demonstrates a concrete implementation of each normal box. Most of the constructions here are given parametrically, using the index of each box. This feature plays an important role in simplifying the search for box relations, as discussed in \cref{sec:normalization-with-relations,app:relations,sec:design}. Also note that the normal box construction uses some circuit shorthands defined in \cref{fig:circuit-conventions}. In \cref{sec:other-definitions}, \cref{fig:Z-Normal-Boxes,fig:X-Normal-Boxes} illustrate the individual construction of each normal box using $H$, $S$, $\CZ$, SWAP, and $\XC$ gates. \cref{fig:auto-normal-box} presents the complete stabilizer tableaus for these boxes.

Using these normal boxes, we now define the normal form for $n$-qutrit Clifford circuits. 

\begin{figure}[!tbh]
    \centering
\scalebox{0.8}{\input{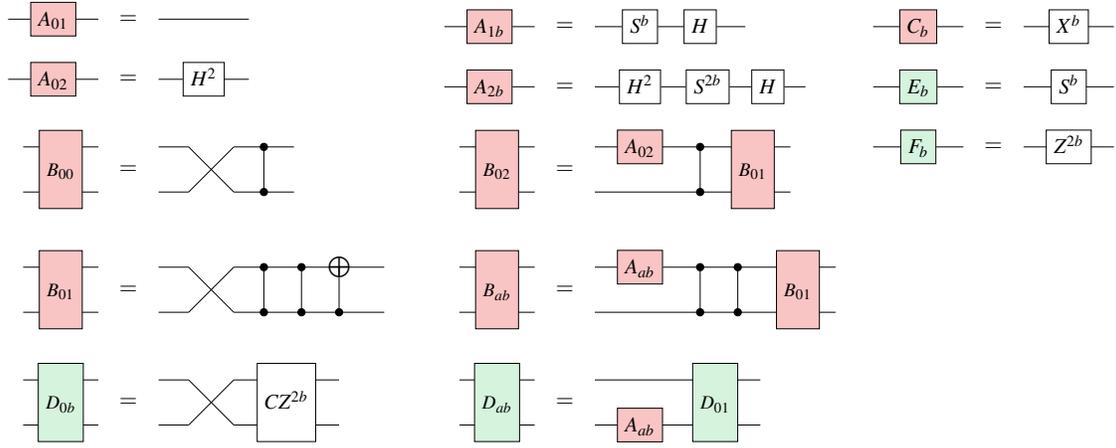}}
    \caption{The concrete implementations of the Z and X normal boxes. Here $a,b\in \Z_3$ and $a\neq 0$.}
    \label{fig:Z-and-X-Normal-Boxes}
\end{figure}

\begin{definition}\label{def:Z-and-X-normal}
For $A$, $B$, $C$, $D$, $E$, and $F$ normal boxes defined in \cref{fig:Z-and-X-Normal-Boxes}:
\begin{itemize}
    \item An $n$-qutrit Clifford circuit is \emph{Z-normal} if it is of the form in \cref{fig:Z-normal}.
    \item An $n$-qutrit Clifford circuit is \emph{X-normal} if it is of the form in \cref{fig:X-normal}.
\end{itemize}
\end{definition}

\begin{figure}[!htb]
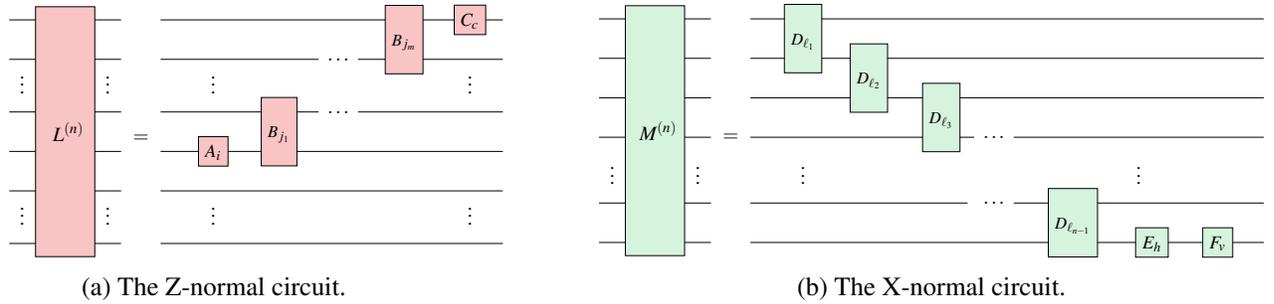

\begin{subfigure}{.35\textwidth}
  \centering
   \scalebox{0.7}{\input{figures/NormalForm/ZNormal.tikz}}
  \caption{The Z-normal circuit.}
  \label{fig:Z-normal}
\end{subfigure}%
\qquad\qquad
\begin{subfigure}{.65\textwidth}
  \centering
   \scalebox{0.7}{\input{figures/NormalForm/XNormal.tikz}}
      \caption{The X-normal circuit.}
    \label{fig:X-normal}
\end{subfigure}
\caption{The Z- and X-normal circuits are used to build the normal form for an $n$-qutrit Clifford circuit.}
\label{fig:normal-components}
\end{figure}

\begin{definition}
An $n$-qutrit Clifford circuit is \emph{normal} if it is of the form in \cref{fig:macro-normal-circuit}.
    \label{def:normal}
\end{definition}

\begin{figure}[!htb]
\[
\scalebox{0.85}{\input{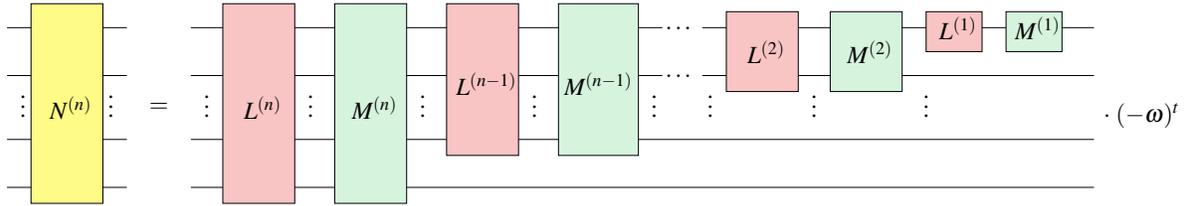}}
\]
\vspace{-.3cm}
    \caption{The structure of a normal form: we alternate layers of Z- and X-normal circuits, with each subsequent layer acting on one fewer qutrit. Here $t \in \Z_6$ defines the global phase. Note that we can view this normal form as being defined inductively, with the $n$-qutrit normal form adding a layer of Z- and X-normal circuits to the left of an existing $(n-1)$-qutrit normal form.}
    \label{fig:macro-normal-circuit}
\end{figure}

\begin{remark}
    When $n=1$, a normal form is composed of only $A$, $C$, $E$, and $F$ boxes. A generic construction of $N^{(1)}$ is given in \cref{fig:single-qutrit-normal}.
\end{remark}

Now that we have the normal form, we need to show that it is \emph{universal}, meaning it can represent an arbitrary Clifford circuit, and \emph{unique}, meaning that each Clifford unitary is uniquely represented by one specific instantiation of this normal form. For universality, by \cref{cor:Clifford-stabilisers}, it suffices to show that we can represent an arbitrary automorphism of $\Pauli_n$ by a normal form.

Let $\phi$ be an automorphism of $\Pauli_n$. We claim that one can construct a $Z$-normal circuit $L$ and an $X$-normal circuit $M$ such that $(ML)^{-1}$ acts as $\phi^{-1}$ on $Z^{(n)}$ and $X^{(n)}$. To obtain a normal form for $\phi$, it suffices to construct $L$ and $M$ and then to recursively iterate this procedure with the updated automorphism $\phi'=\phi(ML)^{-1}$. The proofs of these statements are analogous to those for the qubit case in~\cite{makary2021generators,selinger2015generators}, so we postpone them to \cref{subsec:existence-and-uniqueness} (the only substantial difference is that instead of Paulis $Z$ and $X$ anticommuting with each other, they now `$\omega$-anticommute': $ZX = \omega XZ$).

\begin{lemma}
  \lemznormal
\label{lem:-Z-normal-reduction}
\end{lemma}

\begin{lemma}
  \lemxnormal
\label{lem:-X-normal-reduction}
\end{lemma}

\begin{lemma}
  \lemxnormalpropagateZ
\label{lem:-X-normal-propagate-Z}
\end{lemma}

\begin{lemma}
    \lemnormalcircuitcompo
    \label{lem:-normal-circuit-compositon}
\end{lemma}

\begin{proposition}
\propnormalform
    \label{prop:-normal-form}
\end{proposition}

By the existence statement of \cref{prop:-normal-form}, every automorphism of the Pauli group can be represented as a circuit over $-\omega$, $H$, $S$, and $\CZ$. Hence, all of these automorphisms can be implemented by a Clifford circuit. Conversely, as discussed in \cref{sec:preliminaries}, every Clifford operator is an automorphism of the Pauli group. Hence, \cref{prop:-normal-form} establishes that every Clifford operator admits a unique normal form. This also proves that $\Clifford_n$ is generated by $-\omega$, $H$, $S$, and $\CZ$. Moreover, since there is a bijection between Clifford operators and the normal forms, we can count the number of $n$-qutrit normal forms to compute the cardinality of $\Clifford_n$. The proof of \cref{cor:cardinality} can be found in \cref{subsec:cardinality}.

\begin{corollary}
\corcardinality
    \label{cor:cardinality}
\end{corollary}
\section{A Complete Set of Relations for Multi-Qutrit Clifford Circuits}
\label{sec:normalization-with-relations}

Here, we show how to derive a complete set of rewrite rules for multi-qutrit Clifford circuits by proving that any such circuit can be rewritten into its unique normal form. We call this the \emph{Clifford normalization} process, and it proceeds as follows. Let $C \in \Clifford_n$ be a circuit composed of gates $g_1,\ldots, g_m$. Append it on the right with the normal form of the identity circuit ($N_I$). Starting from the right-most gate $g_m$, iteratively `push' each gate of $C$ into the normal form, updating it at each step (e.g., from $N_I$ to $N'$). This process continues until all gates have been absorbed into the updated normal form $N$. Below, we sketch a simple example, while \cref{fig:normalization1} in \cref{sec:design} provides a detailed illustration of this process.

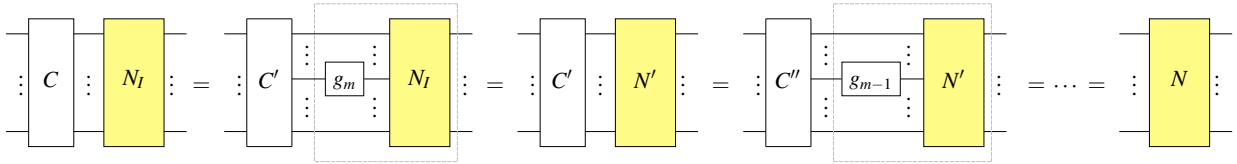
\vspace{-.3 cm}

    \[
\scalebox{.8}{\input{figures/Normalization/normalization.tikz}}
\]

\paragraph{Overview}
In \eqref{eq:identities} and \eqref{eq:n-qutrit-identity-normal}, we present the normal form of the identity operator. Since $\Clifford_n$ is generated by $-\omega$, $H$, $S$ and $\CZ$, it suffices to show how a normal form absorbs these gates and is updated to a new normal form. Locally, this reduces to demonstrating how to push a Clifford gate through each of the normal boxes specified in \cref{fig:Z-and-X-Normal-Boxes}. We call each such instance a \emph{box relation}.

\cref{fig:residual-dirty-gates}.(a) gives an abstract illustration of \emph{pushing through a normal box}, which results in an updated normal box $M'$ and \emph{residual dirty gates}, denoted as $\dir$, appearing after $M'$. $M'$ is uniquely determined by $g$ and $M$. \cref{lem:HA,lem:CZ.B.B-CZ.D.D-updated} provide examples of finding $M'$ and discuss this process in details. $\dir$ denotes a Clifford circuit, and it is determined by $g$, $M$, and $M'$ collectively. \cref{fig:residual-dirty-gates}.(b) shows a concrete example of pushing an $H$ gate through a $B$ box. 

\begin{figure}[!htb]
    \centering
    \[
     \scalebox{1}{\input{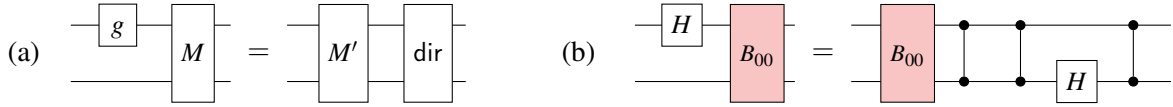}}
    \]
    \vspace{-.3 cm}
    \caption{An abstract and concrete example of pushing a gate through a normal box $M$.}
\label{fig:residual-dirty-gates}
\end{figure}

In total, there are $380$ box relations, which are summarized in \cref{app:relations}. These relations cover every possible case of pushing through a normal box, making them sufficient to normalize any Clifford operator. Combined with \cref{prop:-normal-form}, any two Clifford circuits that are equal as linear maps can be rewritten into the same normal form using these box relations. This shows that the box relations are \emph{complete} for the multi-qutrit Clifford group. In \cite{Supplement2024}, we further prove that these relations can be derived from the $18$ gate relations listed in \cref{fig:rewriterules}. Together with the generating set $\{-\omega, H, S,\CZ\}$ and the derived generators in \cref{fig:summary-derived-generators}, this provides a compact and canonical presentation of the Clifford group by generators and relations.

\paragraph{Technical details}
We need to show that the Clifford normalization terminates, and when this happens, there are no more gates left on the right-hand side of the normal form. \cref{fig:dirty-normal-form} provides a schematic summary of what a normal form could look like during the normalization process. It accounts for all possible cases when pushing through a normal box. 

\begin{definition}\label{def:dirty-normal-form}
We say that a circuit is in \emph{clean normal form} if it is of the form in \cref{fig:macro-normal-circuit}.
A circuit is in \emph{dirty normal form} if it is of the form in \cref{fig:dirty-normal-form}, which has the structure of a normal form, but with any number of additional Clifford gates allowed in all the designated spots, subject to the following rules.
\begin{itemize}
    \item $H$ gates can be present at any wire labelled $\circled{1}$;
     \item $S$ gates can be present at any wire labelled $\circled{1}$,
    $\circled{2}$, $\circled{3}$, or $\circled{4}$;
    \item $X$ gates can be present at any wire labelled $\circled{2}$;
    \item $Z$ gates can be present at any wire labelled $\circled{2}$, $\circled{3}$, $\circled{4}$, or $\circled{5}$;
    \item \CZ gates can be present at any pair of adjacent wires, provided that the top wire is labelled $\circled{1}$, $\circled{2}$, or $\circled{3}$, and the bottom wire is labelled $\circled{1}$.
\end{itemize}
In the context of a dirty normal form, we will refer to the normal boxes as \emph{clean}, and all the other $H$, $S$, $X$, $Z$, and $\CZ$ gates as \emph{dirty}.
\end{definition}

\begin{figure}[!htb]
    \centering
     \scalebox{0.72}{\input{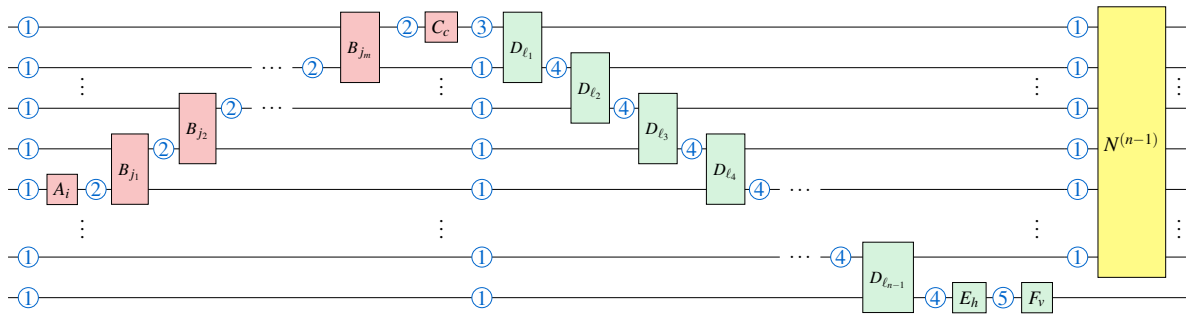}}
    \caption{An example of an $n$-qutrit dirty normal form. Here, $N^{(n-1)}$ is in dirty normal form subject to the same constraints as $N^{(n)}$.}
    \label{fig:dirty-normal-form}
\end{figure}

Note that at some locations, for instance those labelled by $\circled{2}$, a Pauli $X$ or $Z$ gate is allowed to appear in the dirty normal form. However, an $H$ gate is not allowed at $\circled{2}$, even though it is used to construct Pauli gates (see \cref{fig:summary-derived-generators}). The annotation in \cref{fig:dirty-normal-form} emphasizes that there can be no solitary $H$ gates at the locations labelled by $\circled{2}$, $\circled{3}$, $\circled{4}$, or $\circled{5}$. It is, however, possible for an $H$ gate to appear as `a part of' other Clifford gates. The same convention applies to the $S$ gate.

The placement of dirty gates is restricted by which gates can be properly absorbed by each normal box, as this ultimately determines whether Clifford normalization will terminate. \Cref{fig:dirty-gate-restrictions-1,fig:dirty-gate-restrictions-2,fig:dirty-gate-restrictions-3} summarize the allowed Clifford operators that may precede the $B$, $C$, $D$, $E$, and $F$ boxes. Accordingly, the third column of \cref{fig:normal-box-auto} identifies the additional symmetry each normal box must satisfy to ensure dirty gates always have the expected shape. By manually inspecting the box relations in \cref{app:relations}, we can confirm that all dirty gates follow the pattern described in \cref{def:dirty-normal-form}.

\begin{lemma}
    Any dirty normal form can be converted to its normal form by applying the box relations in \cref{app:relations}, in the left-to-right direction, a finite number of times.
    \label{lem:conversion}
\end{lemma}

\begin{proof}
We adapt the proof of Lemma 5.4 in \cite{makary2021generators}. By \cref{def:dirty-normal-form}, dirty gates always appear before clean ones. Hence, if a dirty gate remains in the circuit, it must occur immediately before a clean normal box. The left-hand side of the rewrite rules in \cref{app:relations} contains all cases of a dirty gate occurring immediately before a clean normal box. Thus, so long as there are dirty gates in the circuit, one of the rules in \cref{app:relations} can be applied. Moreover, every application of a rule maps a dirty normal form to a dirty normal form (for this it is important to check that the residual dirty gates in our relations satisfy the constraints outlined in \cref{def:dirty-normal-form}).

We now show that this process ends in finitely many rewrites. For this, we assign, to each dirty normal form, a sequence of nonnegative integer numbers. Assume that a dirty normal form has $t$ clean gates (i.e.~normal boxes), numbered $1,\ldots,t$ from left to right. Define $s =(s_1,\ldots,s_t)$ where $s_i$ is the number of dirty gates that occur before the $i$-th clean gate. A left-to-right application of one of the box relations at the clean gate $i$ decreases $s_i$ and leaves $s_1,\ldots, s_{i-1}$ invariant so that $s$ decreases in the lexicographic order. The length of the sequence of clean gates is not guaranteed to remain constant through the normalization process, but it is bounded above by the maximum possible number of clean gates. For a normal form on $n$ qutrits, this bound is given by $\sum_{i=1}^{n} 2(i+1) = n^{2}+3n$.
  Therefore, this procedure terminates after finitely many rewrites.
\end{proof}

\begin{proposition}
    Consider a Clifford circuit $C \in \Clifford_n$ expressed in terms of the generators $-\omega$, $H$, $S$, and \CZ gates. Any such circuit can be converted to its normal form by using the box relations in \cref{app:relations}, together with the equations:
    
    \begin{equation}
    \scalebox{.9}{\input{figures/Normalization/identities.tikz}} \label{eq:identities}
  \end{equation}
    \label{prop:normalization}
\end{proposition}

\begin{proof}
For $1 \leq k \leq n$, let $N_{I_k}$ be the normal form of the $k$-qutrit identity operator. First note that 
\begin{equation}\label{eq:n-qutrit-identity-normal}
    \scalebox{.7}{\input{figures/Normalization/n-qutritIdentityNormal.tikz}}.
\end{equation}

Indeed, the identity circuit can be converted to this normal form by applying \eqref{eq:identities} a finite number times. Appending the given Clifford circuit $C$ on the right with the normal form in \eqref{eq:n-qutrit-identity-normal}, we obtain a dirty normal form, which can be converted to its normal form by \cref{lem:conversion}.
\end{proof}

Together, \cref{prop:-normal-form,prop:normalization} show that qutrit Clifford operators are presented by the generators $-\omega$, $H$, $S$, and $\CZ$, the derived generators defined in \cref{fig:summary-derived-generators}, the normal boxes given in \cref{fig:Z-and-X-Normal-Boxes}, as well as the $380$ box relations listed in \cref{app:relations}. This presentation is highly redundant, so we provide a reduced set of relations as listed in \cref{fig:rewriterules}. This rule set remains complete for multi-qutrit Clifford circuits, but it offers a more intuitive framework for us to understand the qutrit Clifford gate interactions.

\begin{theorem}
    The $18$ rewrite rules listed in \cref{fig:rewriterules} are complete for qutrit Clifford circuits.
    \label{prop:reduced-relations}
\end{theorem}

\begin{proof}
It suffices to show that these $18$ relations imply the $380$ box relations. Proofs can be found in the supplement to this paper~\cite{Supplement2024}.
\end{proof}

\section{Conclusion}
\label{sec:conclusion}

Generalizing from qubits to qudits, we present the first completeness result for the multi-qudit Clifford group. We consider the case when $d=3$, and generalize the methods from \cite{makary2021generators,selinger2015generators} to the qutrit setting. This is done in three steps. First, we define a normal form and show that every qutrit Clifford circuit can be uniquely written in this form. Then, we show how each Clifford circuit can be reduced to this normal form and find a set of $380$ box relations that suffice for this process. Finally, we prove that these box relations can be reduced to the $18$ gate relations in \cref{fig:rewriterules}. This complete set of rewrite rules offers a clear and intuitive framework for understanding the qutrit Clifford gate interactions. 

Notably, our generalization of the normal form from the qubit setting is far from obvious. In~\cite{makary2021generators,selinger2015generators}, the design space of the normal form is small enough that a working decomposition of the normal form could be found through brute force. In the qutrit case, however, the increased degrees of freedom make this approach significantly more tedious and labour-intensive. We therefore adapt the previous techniques to account for the subtleties of qutrit arithmetic and operations. This involves identifying additional symmetries that the normal form should satisfy, leading to a more canonical choice of normal form components.

Several open questions remain, the first of which concerns whether the reduced rule set presented in \cref{fig:rewriterules} is minimal. We can show that all the single-qutrit Clifford relations except for (C5), the $S$-$S'$ commutation rule, are necessary. Thus, (C1) to (C5) are nearly minimal. However, we are not aware of any arguments for the multi-qutrit Clifford relations. (C6) to (C18) appear quite natural, making it difficult to determine which, if any, could be eliminated. Alternatively, a more compact rule set might be achievable by using rotation operators as the generators for the qutrit Clifford group. Building on these considerations, we plan to generalize the framework developed in this paper to Clifford circuits in other odd prime dimensions. This will require more general techniques for constructing normal forms, supported by software tools for deriving, reducing, and verifying relations.

\paragraph{Acknowledgements}
The authors would like to thank Xiaoning Bian, Maris Ozols, and Peter Selinger for enlightening discussions, as well as the anonymous reviewers of the 22nd International Conference on Quantum Physics and Logic (QPL 2025) for their insightful comments on an earlier version of this paper. SML would like to thank Boldizsár Poór, Razin A. Shaikh, and Lia Yeh for their valuable feedback and support. The circuit diagrams in this paper and its supplement~\cite{Supplement2024} were typeset using TikZiT~\cite{tikz}. 

SML and MM thank NTT Research for their financial support. This work was supported in part by Canada’s NSERC. Research at IQC is supported in part by the Government of Canada through Innovation, Science and Economic Development Canada. Research at Perimeter Institute is supported in part by the Government of Canada through the Department of Innovation, Science and Economic Development Canada and by the Province of Ontario through the Ministry of Colleges and Universities. YZ is supported by VILLUM FONDEN via QMATH Centre of Excellence grant number 10059 and Villum Young Investigator grant number 37532.

\bibliographystyle{eptcs}
\bibliography{qutrit}

\appendix

\clearpage
\newgeometry{margin=2cm}
\section{Complementary Definitions}
\label{sec:other-definitions}

Here we provide complementary definitions that are helpful for discussing our results. 

By putting \cref{def:Z-and-X-normal,def:normal} together, \cref{fig:micro-normal-circuit} shows the normal form of an $n$-qutrit Clifford circuit when expanding the Z- and X-normal circuits in \cref{fig:normal-components}. \cref{fig:Z-Normal-Boxes,fig:X-Normal-Boxes} describe the individual construction of each normal box in terms of $H$, $S$, $\CZ$, and $\XC$ gates.



According to \cref{subsec:action}, we can describe the automorphism of each normal box by its actions on the Pauli generators. That is, consider a normal box $M$ and a Pauli generator $P$, $MPM^\dagger = Q$, for some Pauli $Q$. For our purposes, it is also convenient to describe $M$'s \emph{inverse conjugation} on Pauli generators: $M^\dagger PM = Q'$, for some Pauli $Q'$.

For all normal boxes, \cref{fig:Z-and-X-Normal-Boxes} provides a high-level overview of their constructions using parameterized definitions. \cref{fig:Z-Normal-Boxes,fig:X-Normal-Boxes} provide a microscopic picture of their constructions. Based on these two perspectives, \cref{fig:auto-normal-box} summarizes important behaviours of normal boxes. The action of $A$, $C$, $E$, and $F$ boxes on $\Pauli_1$ are described by their actions on Paulis $X$ and $Z$, where we write the (pre)image of $X$ above the qutrit wire, and that of $Z$ below it. Based on the generic parameterized constructions of $B$ and $D$ boxes, it suffices to specify the automorphism of $B_{00}$, $B_{01}$, $D_{00}$, $D_{01}$, and $D_{02}$. Also note that $D_{02} = B_{00}$ and $D_{01} = B^{-1}_{00}$.

\begin{figure}[H]
    \begin{equation*}
        \scalebox{0.75}{\input{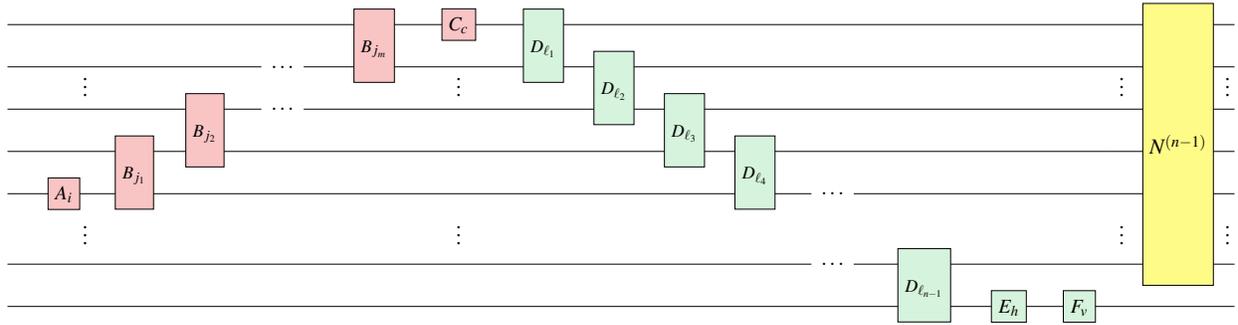}}
    \end{equation*}
    \caption{A close-up view of the normal form of an $n$-qutrit Clifford circuit. $N^{(n-1)}$ has the same structure as the circuit before it, but on one fewer qutrit.} 
    \label{fig:micro-normal-circuit}
\end{figure}

\clearpage
\restoregeometry

\begin{figure}[H]
    \centering
\scalebox{.9}{\input{figures/Alternatives/ZNormalBoxes.tikz}}
    \caption{A microscopic picture of the $Z$ normal boxes.}
    \label{fig:Z-Normal-Boxes}
\end{figure}

\begin{figure}[H]
    \centering
\scalebox{.9}{\input{figures/Alternatives/XNormalBoxes.tikz}}
    \caption{A microscopic picture of the $X$ normal boxes.}
    \label{fig:X-Normal-Boxes}
\end{figure}

 \begin{figure}[H]
     \centering
      \[
\scalebox{0.8}{\input{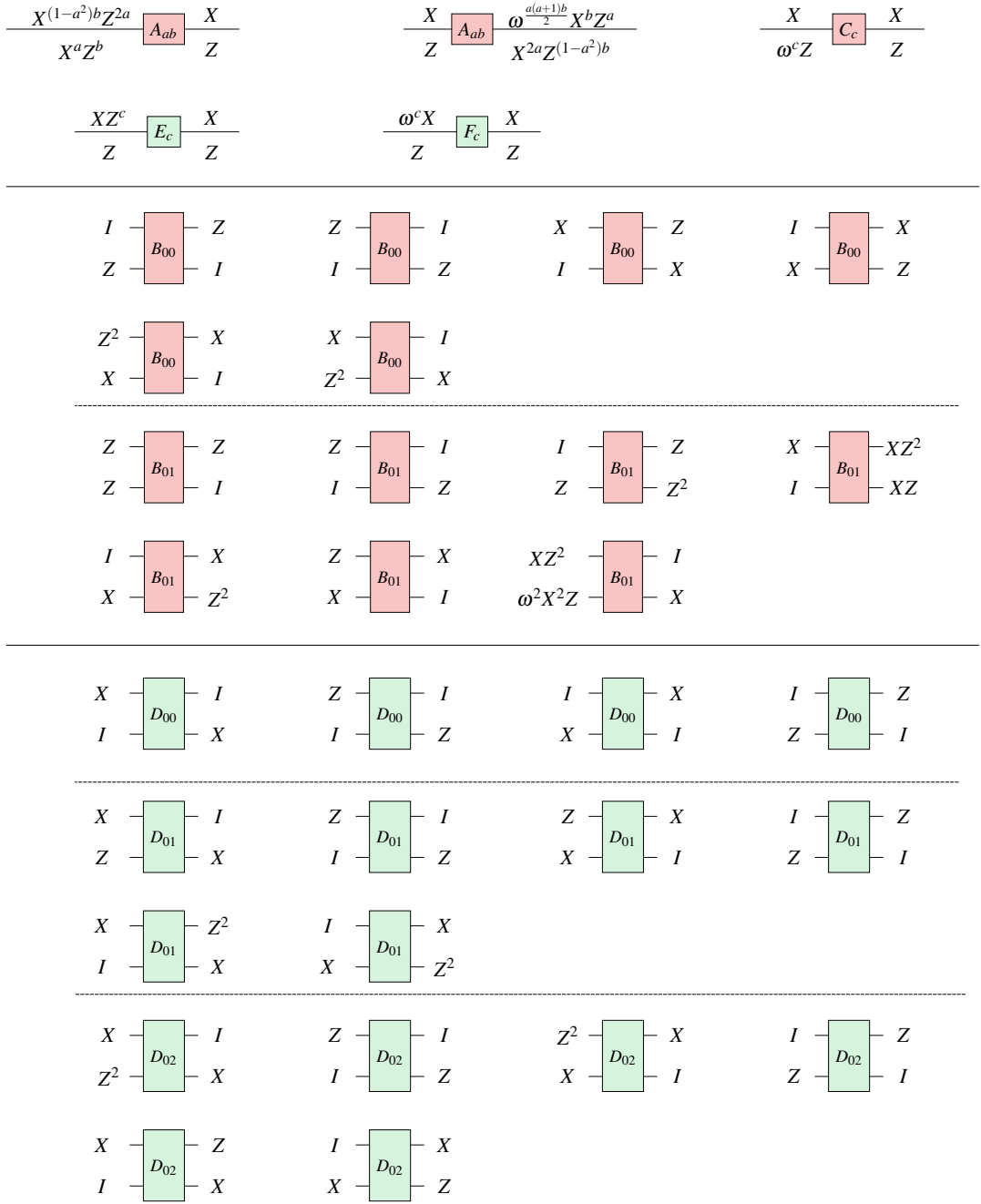}}
\]
     \caption{The actions of $A$, $C$, $E$, and $F$ boxes on $\Pauli_1$ are described by their actions on the single-qutrit Pauli generators $X$ and $Z$. $(a,b) \in \Z_3 \times \Z_3\setminus \{(0,0)\}$ and $c \in \Z_3$. The actions of $B_{00}$, $B_{01}$, $D_{00}$, $D_{01}$, and $D_{02}$ on $\Pauli_2$ are described by their actions on the two-qutrit Pauli generators $X\otimes I$, $I\otimes X$, $Z\otimes I$, and $I\otimes Z$. For each Pauli generator, we display both their images and preimages within a normal box. These will be useful for finding all box relations.}
     \label{fig:auto-normal-box}
 \end{figure}

\newpage
\section{Complementary Proofs}
\label{sec:proofs}
Here, we provide detailed proofs of the results presented in \cref{sec:normal-form}. Let $\lambda \in \C$. $\lambda^\dagger$ denotes the complex conjugate of $\lambda$. When $u$, $\tau$, and $\nu$ are elements of a ring $R$, we write $u\equiv_\nu \tau$ if $u$ is congruent to $\tau$ modulo $\nu$. That is, $u-\tau= r\nu$, for some $r \in R$. We write  $u\not\equiv_\nu \tau$ otherwise.

\subsection{Scalars of \texorpdfstring{$\Clifford_n$}{Cn}}
\label{subsec:scalars}

We start by showing that all Clifford scalars are of the form $(-\omega)^t$, for $t \in \Z_6$. This becomes useful for finding the cardinality of $\Clifford_n$, as shown in \cref{cor:cardinality}. To this end, we first show that the phases (complex numbers of unit norm) that belong to the ring $\Tw$ are exactly the complex numbers of the form $(-\omega)^t$, for $t \in \Z_6$. Since the complex numbers of this form all belong to $\Tw$, we only need to show that if an element of $\Tw$ has norm 1, then it is of the required form.

\begin{lemma}
    Let $a,b\in\Z$. If $a\not\equiv_3 0$, $b\not\equiv_3 0$, and $a\not\equiv_3 b$, then $a^2 -ab + b^2 \equiv_9 3$.
    \label{lem:phases}
\end{lemma}

\begin{proof}
Since $a\not\equiv_3 0$, $b\not\equiv_3 0$, and $a\not\equiv_3 b$, we have either $a\equiv_3 1$ and $b\equiv_3 2$, or $a\equiv_3 2$ and $b\equiv_3 1$. Regardless, we then have
\[
a+b\equiv_3 0
\quad
\mbox{ and }
\quad
ab \equiv_3 2.
\]
The above implies that
\[
(a+b)^2\equiv_9 0
\quad
\mbox{ and }
\quad
3ab \equiv_9 6.
\]
Indeed, if $a+b=3\ell$ for some integer $\ell$, then $(a+b)^2=9\ell^2$; and similarly, if $ab=2+3\ell$, then $3ab = 6+9\ell$. Hence, we get
\[
a^2 - ab +b^2 = a^2 + 2ab + b^2 - 3ab = (a+b)^2 - 3ab \equiv_9 0 - 6 \equiv_9 3,
\]
as desired.
\end{proof}

\begin{proposition}
    Let $x\in\Tw$. If $||x|| =1$, then $x=(-\omega)^t$, for $t \in \Z_6$.
    \label{prop:phases}
\end{proposition}

\begin{proof}
Let $x\in \Tw$. Then $x$ can be written as $x = (a+\omega b)/3^k$, for some integers $a$, $b$, and $k$, with $k\geq 0$ and $k$ minimal. Since $||x||=1$, we then have $x^\dagger x = 1$, and thus, 
\begin{equation}
\label{eq:phases1}
(a+\omega b)^\dagger (a+\omega b) = 3^{2k}.
\end{equation}
Since $\omega^\dagger = \omega^2 = -1-\omega$, we can rewrite \eqref{eq:phases1} as
\begin{equation}
\label{eq:phases2}
3^{2k} = (a+\omega b)^\dagger (a+\omega b) = a^2 + \omega ab + \omega^\dagger ab + b^2 = a^2 -ab + b^2.
\end{equation}

If $k=0$, then \eqref{eq:phases2} becomes $a^2 - ab + b^2 =1$, whose only solutions are $a=\pm 1$ and $b=0$, $a=0$ and $b=\pm 1$, and $a = \pm 1$ and $b = \mp 1$. These six solutions correspond exactly to $x=(-\omega)^t$, for $t \in \Z_6$.

To complete the proof, we now show that there are no solutions to \eqref{eq:phases2} with $k\geq 1$. It suffices to show that $k\geq 1$ implies $a\not\equiv_3 0$, $b\not\equiv_3 0$, and $a\not\equiv_3 b$. Indeed, if this is the case, then \cref{lem:phases} shows that $a^2 -ab +b^2$ is not congruent to 0 modulo 9, which contradicts \eqref{eq:phases2}, since $3^{2k}\equiv_9 0$ when $k\geq 1$. If $k\geq 1$, then $a\not\equiv_3 0$ or $b\not\equiv_3 0$, since both $a$ and $b$ being divisible by $3$ would contradict the minimality of $k$. But if $a\equiv_3 0$, then \eqref{eq:phases2} implies that $b^2 \equiv_3 0$, so that $b\equiv_3 0$. Similarly, if $b\equiv_3 0$, then $a\equiv_3 0$. Therefore, it must be the case that both $a\not\equiv_3 0$ and $b\not\equiv_3 0$. Finally, since neither $a$ nor $b$ is congruent to 0 modulo 3, $a\equiv_3 b$ would imply
\[
a^2 -ab + b^2 \equiv_3 1 -1 + 1 \equiv_3 1,
\]
which contradicts \eqref{eq:phases2}. Thus we must also have $a\not\equiv_3 b$.
\end{proof}

\begin{corollary}
    Let $U, V \in \Clifford_n$ and $U, V$ act identically on $\Pauli_n$. Then $V = (-\omega)^t U$, $t \in \Z_6$.
    \label{cor:phases}
\end{corollary}

\begin{proof}
For all $P \in \Pauli_n$, since $UPU^\dagger = VPV^\dagger$, $(V^\dagger U)P(U^\dagger V) = (V^\dagger V)P(V^\dagger V) = P$. By \cref{prop:scalars}, $V^\dagger U = \lambda$, for some $\lambda \in \C$. Recall that $\Clifford_n = \langle H, S, \CZ \rangle$. By direct computation, $H$, $S$, and $\CZ$ are unitaries over $\Z[1/3, \omega]$. This implies that $U,V \in \Clifford_n$ are unitaries over $\Z[1/3, \omega]$. Hence $\lambda \in \Z[1/3, \omega]$ and $\lambda \lambda^\dagger = 1$. By \cref{prop:phases}, $\lambda \in \{(-\omega)^t$;\; $t \in \Z_6\}$. Based on the relation (C4) in \cref{fig:rewriterules}, $(HS^2)^3 = -\omega$. Hence, $\Clifford_n$ must contain all powers of $-\omega$. This completes the proof.
\end{proof}

\subsection{Proofs of Lemmas \ref{lem:-Z-normal-reduction}, \ref{lem:-X-normal-reduction}, \ref{lem:-X-normal-propagate-Z}, \ref{lem:-normal-circuit-compositon}, and Proposition \ref{prop:-normal-form}}
\label{subsec:existence-and-uniqueness}
Following the conventions in \cref{subsec:action}, we use $C \bullet P$ to denote the conjugation of a Clifford $C \in \Clifford_n$ on a Pauli $P \in \Pauli_n$. Let $Q=C \bullet P$. We say that in $C$, $Q$ is the \emph{image} of $P$ and $P$ is the \emph{preimage} of $Q$. Since $C$ is unitary, we can also take the inverse of the circuit to get $C^\dagger \bullet Q = C^\dagger Q C = P$. 

\begin{definition}
    Let $P, Q \in \Pauli_n$ and $PQ = \omega QP$. Then we say $P$ \emph{$\omega$-anticommutes} with $Q$.
   \label{def:omega-anticommute}
\end{definition}

\begin{T1}
  \lemznormal
\end{T1} 
\begin{proof}
By \cref{def:Pauli-groups}, an $n$-qutrit Pauli operator $P$ has the form $P = \omega^c \bigotimes_{j=1}^{n}P_j$, $P_j = X^{a_j}Z^{b_j}$, for $a_j, b_j, c \in \Z_3$. Since $P \notin \{I, \omega I, \omega^2 I\}$, there exists $m \in \N$ and $1\leq m \leq n$ such that $P_m$ is not an identity. Let $m$ be the largest such index, then

\begin{equation}
    P = (\omega^cP_1) \otimes \cdots \otimes P_m \otimes I\otimes \cdots \otimes I = (\omega^cX^{a_1}Z^{b_1})\otimes \cdots \otimes (X^{a_m}Z^{b_m}) \otimes I \otimes \cdots \otimes I.
    \label{eq:P-decomposition}
\end{equation}

Next, we show how to construct a unique $Z$-normal circuit $L$ based on \eqref{eq:P-decomposition} and the unique normal box actions specified in the second column of \cref{fig:normal-box-auto}. Starting from qutrit wire $m$, $A_{a_m,b_m}$ is uniquely determined by $P_m = X^{a_m}Z^{b_m}$. As a result, this $A$ box maps $X^{a_m}Z^{b_m}$ to $Z$. Consequently, $B_{a_{m-1},b_{m-1}}$ is uniquely determined by $P_{m-1} \otimes Z =  X^{a_{m-1}}Z^{b_{m-1}} \otimes Z$. Then, this $B$ box maps $P_{m-1} \otimes Z$ to $Z \otimes I$. Continue this process by concatenating $B$ normal boxes upward, until reaching the top two qutrit wires, where $B_{a_1,b_1}$ is uniquely determined by $\omega^cP_1\otimes Z=\omega^c X^{a_{1}}Z^{b_{1}}\otimes Z$. As a result, $\omega^cP_1\otimes Z$ is mapped to $\omega^c Z \otimes I$. Finally, $C_{c}$ is uniquely determined by $\omega^c Z$, and it maps $\omega^c Z$ to $Z$. This gives us a $Z$-normal circuit $L$ that is displayed below. It is uniquely determined by $P$ and $L\bullet P = Z \otimes I \otimes \cdots \otimes I$.

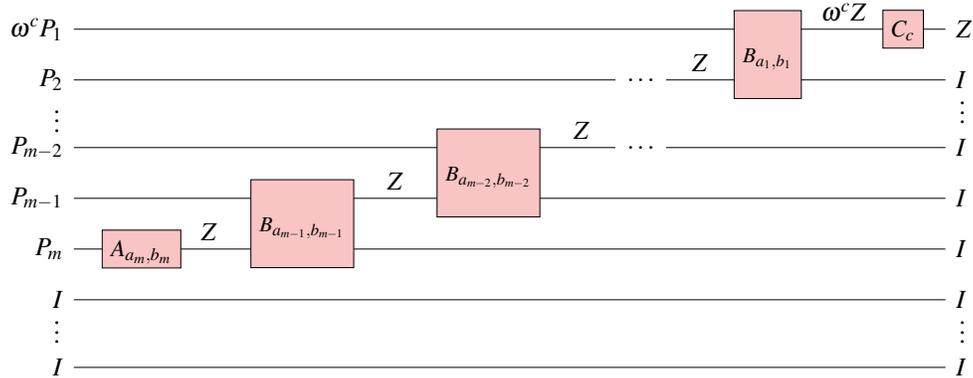
\[
\scalebox{.9}{\input{figures/NormalForm/ZNormalReduction.tikz}}\qedhere
\] 
\end{proof}

Given a $Z$-normal circuit $L$, we can read off its preimage of $Z \otimes I \otimes \cdots \otimes I$ by reading the indices of each normal box in $L$ from right to left.

\begin{T2}
  \lemxnormal
\end{T2}   

\begin{proof}
By \cref{def:Pauli-groups}, an $n$-qutrit Pauli operator $Q$ has the form $Q = \omega^c \bigotimes_{j=1}^{n}Q_j$, $Q_j = X^{a_j}Z^{b_j}$, for $a_j, b_j, c \in \Z_3$. Since $(Z \otimes I \otimes \cdots \otimes I)Q = \omega Q (Z \otimes I \otimes \cdots \otimes I)$, $Q_1 =XZ^i$, for some $i \in \Z_3$. Then

\begin{equation}
    Q= Q_1 \otimes Q_2 \otimes \cdots \otimes (\omega^cQ_n)=(XZ^i) \otimes (X^{a_2}Z^{b_2}) \otimes \cdots \otimes (\omega^cX^{a_n}Z^{b_n}).
    \label{eq:Q-decomposition}
\end{equation}

Next, we show how to construct a unique $X$-normal circuit $M$ based on \eqref{eq:Q-decomposition} and the unique normal box actions specified in the second column of \cref{fig:normal-box-auto}. Starting from qutrit wires $1$ and $2$, $D_{a_2,b_2}$ is uniquely determined by $Q_1 \otimes Q_2 = (XZ^i) \otimes (X^{a_2}Z^{b_2})$. As a result, this $D$ box maps $(XZ^i) \otimes Q_2$ to $I \otimes (XZ^i)$. Consequently, $D_{a_{3},b_{3}}$ is uniquely determined by $(XZ^i) \otimes Q_3 =  (XZ^i) \otimes (X^{a_3}Z^{b_3})$. Again, this $D$ box maps $(XZ^i) \otimes Q_3$ to $I \otimes (XZ^i)$. Continue this process by concatenating $D$ normal boxes downward, until reaching the bottom two qutrit wires, where $D_{a_n,b_n}$ is uniquely determined by $(XZ^i) \otimes \omega^cQ_n =  (XZ^i) \otimes (\omega^cX^{a_n}Z^{b_n})$. Thus, $(XZ^i) \otimes \omega^cQ_n$ is mapped to $I \otimes (\omega^c XZ^i)$ by this last $D$ box.  Then, $E_i$ is uniquely determined by $XZ^i$, and it maps $\omega^cXZ^i$ to $\omega^c X$. Finally, $F_c$ is uniquely determined by $\omega^c X$, and it maps $\omega^c X$ to $X$. This gives us an $X$-normal circuit $M$ that is displayed below. It is uniquely determined by $Q$ and $M\bullet Q = I \otimes \cdots \otimes I \otimes X$. 

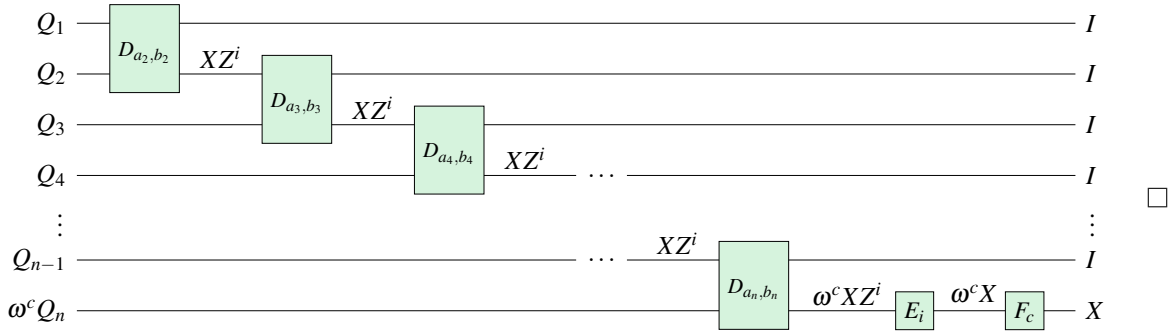
\[
\scalebox{0.9}{\input{figures/NormalForm/XNormalReduction.tikz}}\qedhere
\]
\end{proof}

Given an $X$-normal circuit $M$, we can read off its preimage of $I \otimes \cdots \otimes I \otimes X$ by reading the indices of each normal box in $M$ from right to left.

\begin{T3}
  \lemxnormalpropagateZ
\end{T3}   

\begin{proof}
This follows from the additional actions of $D$, $E$, and $F$ boxes on certain Pauli operators (see the third column of \cref{fig:normal-box-auto}). In particular, every $D$ box maps $Z \otimes I$ to $I \otimes Z$, every $E$ and $F$ box maps $Z$ to $Z$. It is visualized below.
    \[
    \scalebox{0.9}{\input{figures/NormalForm/XNormalZCommute.tikz}}\qedhere
    \]
\end{proof}

Intuitively, an $X$-normal circuit $M$ propagates a Pauli $Z$ from the top wire to the bottom wire.

\begin{T4}
  \lemnormalcircuitcompo
\end{T4}  
\begin{proof}
    Since $PQ = \omega QP$, $P \notin \{I, \omega I, \omega^2 I\}$. By \cref{lem:-Z-normal-reduction}, there exists a unique $Z$-normal circuit $L$ such that $L \bullet P = Z \otimes I \otimes \cdots \otimes I$. Moreover, $(L\bullet P) (L\bullet Q) = L\bullet (PQ) = L\bullet (\omega QP) = \omega (L\bullet Q)  (L\bullet P)$. Hence $(Z \otimes I \otimes \cdots \otimes I) (L\bullet Q) = \omega (L\bullet Q)  (Z \otimes I \otimes \cdots \otimes I)$. By \cref{lem:-X-normal-reduction}, there exists a unique $X$-normal circuit $M$ such that $M \bullet (L\bullet Q) = I \otimes \cdots \otimes I \otimes X$. By \cref{lem:-X-normal-propagate-Z}, $M\bullet(L\bullet P) = I \otimes \cdots \otimes I \otimes Z$. It follows that
    \begin{equation}
        ML\bullet P = I \otimes \cdots \otimes I \otimes Z, \qquad  ML\bullet Q = I \otimes \cdots \otimes I \otimes X.
        \label{eq:condition}
    \end{equation}
    
    This proves the existence of the desired $Z$-normal circuit $L$ and $X$-normal circuit $M$. Now suppose towards contradiction that there exist another $Z$-normal circuit $L'$ and $X$-normal circuit $M'$ satisfying  \eqref{eq:condition}. Since $M'\bullet (L'\bullet P) = I \otimes \cdots \otimes I \otimes Z$, \cref{lem:-X-normal-propagate-Z} implies that $L'\bullet P = Z\otimes I \cdots \otimes I$. By the uniqueness of the $Z$-normal circuit in \cref{lem:-Z-normal-reduction}, $L=L'$. Then $M'\bullet(L'\bullet Q) = M'\bullet(L\bullet Q)=I \otimes \cdots \otimes I \otimes X$. By the uniqueness of the $X$-normal circuit in \cref{lem:-X-normal-reduction}, $M'=M$.
\end{proof}

\begin{T5}
\propnormalform
\end{T5}

\begin{proof}
We proceed by induction on $n$. When $n=0$, the Pauli operators are scalars of the form $\omega^t$, $t \in \Z_3$. In this case, $\phi$ is the identity. Setting $C = 1$, uniqueness up to scalar follows from the fact that when $n=0$, the Clifford operators are scalars of the form $(-\omega)^t, t\in \Z_6$.

Now suppose that our claim is true for $n-1$ and consider the case of $n$. First we prove existence. Since $\phi$ is bijective, there exist $P, Q \in \Pauli_n$ such that $\phi(P) = I \otimes \cdots \otimes I \otimes Z$ and $\phi(Q) = I \otimes \cdots \otimes I \otimes X$. That is, $P = \phi^{-1}(I \otimes \cdots \otimes I \otimes Z)$ and $Q =\phi^{-1}(I \otimes \cdots \otimes I \otimes X)$. Then $PQ =\phi^{-1}(I \otimes \ldots \otimes I \otimes ZX)$. Since $I \otimes \ldots \otimes I \otimes Z$ $\omega$-anticommutes with $I \otimes \ldots \otimes I \otimes X$, $PQ=\omega QP$. By \cref{lem:-normal-circuit-compositon}, there exist a unique $X$-circuit $M$ and a unique $Z$-circuit $L$ such that $ ML \bullet P = I \otimes \ldots \otimes I \otimes Z = \phi(P)$ and $ML \bullet Q = I \otimes \ldots \otimes I \otimes X = \phi(Q)$. Let $\phi ': \Pauli_{n} \rightarrow \Pauli_{n}$ be a new automorphism and $\phi'$ fixes scalars.
\begin{equation}
    \phi '(U) = \phi\left((ML)^{-1} \bullet U\right),
    \label{eq:new-auto}
\end{equation}

Then $I \otimes \cdots \otimes I \otimes Z$ and $I \otimes \cdots \otimes I \otimes X$ are fixed points of $\phi'$, since for $T\in \{X,Z\}$,

\begin{align*}
\phi '(I \otimes \dots \otimes I \otimes T) &= \phi\left((ML)^{-1} \bullet (I \otimes \dots \otimes I \otimes T)\right) \\
&= \phi\left((ML)^{-1} \bullet (ML) \bullet T\right) = \phi(T) = I \otimes \dots \otimes I \otimes T.
\end{align*}

For any $R \in \Pauli_{n-1}$, since $R \otimes I$ commutes with $I \otimes \dots \otimes I \otimes Z$ and $I \otimes \dots \otimes I \otimes X$, $\phi '(R \otimes I)$ commutes with $\phi '(I \otimes \dots \otimes I \otimes Z)= I \otimes \dots \otimes I \otimes Z$ and $\phi '(I \otimes \dots \otimes I \otimes X)= I \otimes \dots \otimes I \otimes X$. This implies that $\phi '(R\otimes I) = S \otimes I$ for some $S \in \Pauli_{n-1}$. Then there exists an automorphism $\phi '': \Pauli_{n-1} \rightarrow \Pauli_{n-1}$ such that $\phi''(R) = S$. Since $\phi '$ fixes $I \otimes \dots \otimes I \otimes Z$ and $I \otimes \dots \otimes I \otimes X$, 

\begin{equation}
    \phi ' = \phi '' \otimes I.
    \label{eq:equality}
\end{equation}

By the induction hypothesis, for all $R \in \pauli{n-1}$, there exists $C' \in \Clifford_{n-1}$ in normal form such that 

\begin{equation}
    C' \bullet R = \phi ''(R),
    \label{eq:IH}
\end{equation}

Then $C = (C'\otimes I)  ML$ is in normal form. Next, we show that $C \bullet U = \phi(U)$ for all $U \in \Pauli_n$. By  \eqref{eq:new-auto},

\begin{equation}
    (ML)^{-1}\bullet U = \phi^{-1}  \phi'(U).
    \label{eq:map1}
\end{equation}

By taking the inverse of both sides of \eqref{eq:map1}, we have

\begin{equation}
    ML = (\phi^{-1}  \phi')^{-1} = \phi'^{-1}   \phi.
    \label{eq:map2}
\end{equation}

It follows that 
\[
 C \bullet U =((C' \otimes I)  ML) \bullet U \xlongequal[]{\eqref{eq:map2}} ((C' \otimes I)\bullet\left(\phi'^{-1}\left(\phi(U)\right)\right)\xlongequal[]{\eqref{eq:equality}} (C' \otimes I)\bullet (\phi''^{-1}\otimes I)\left(\phi(U)\right)\xlongequal[]{\eqref{eq:IH}} \phi(U).
\]

To prove uniqueness, suppose towards contradiction that $D \in \Clifford_n$ is another Clifford circuit in normal form such that $D \bullet U = \phi(U)$ for all $U \in \pauli{n}$.  By induction, $D = (D' \otimes I)  M'L'$, where $M'$ is an $X$-normal circuit, $L'$ is a $Z$-normal circuit, and $D'$ is a normal Clifford circuit on $n-1$ qutrits. Since $D \bullet P = \phi(P) = I \otimes \dots \otimes I \otimes Z$, $((D' \otimes I)  (M'L'))\bullet P = I \otimes \dots \otimes I \otimes Z$. Then

\[
  (M'L') \bullet P= (D' \otimes I)^{-1}\bullet(I \otimes \dots \otimes I \otimes Z) = (D'^{-1} \otimes I) \bullet(I \otimes \dots \otimes I \otimes Z)=I \otimes \dots \otimes I \otimes Z.
\]

Similarly, $D\bullet Q = I \otimes \dots \otimes I \otimes X$ implies that $(M'L') \bullet Q = I \otimes \dots \otimes I \otimes X$. By the uniqueness of the $X$- and $Z$-normal circuits in \cref{lem:-normal-circuit-compositon}, $M'=M$ and $L'=L$. By induction hypothesis, $C'=D'$, up to $(-\omega)^t, t\in \Z_6$. Thus, the same is true for $C$ and $D$. This completes the proof.
\end{proof}

\subsection{Proof of Corollary \ref{cor:cardinality}}
\label{subsec:cardinality}

\begin{T6}
  \corcardinality
\end{T6}   

\begin{proof}
By the uniqueness of the normal form and the correspondence established in \cref{prop:-normal-form}, it is sufficient to count the total number of distinct normal forms. Up to scalar, consider the normal form of an arbitrary $n$-qutrit Clifford operator:

 \[
    \input{figures/NormalForm/cardinality2.tikz}
    \]

By induction, our problem is reduced to counting different ways of constructing $L^{(n)}$ and $M^{(n)}$.
\begin{itemize}
    \item For $L^{(n)}$, according to \cref{fig:Z-and-X-Normal-Boxes}, there are $8$ choices for an $A$ box and $3$ choices for a $C$ box. For $B$ boxes, we can start concatenating them upward from any qutrit wire. According to \cref{fig:Z-normal}, we proceed by cases.
    \begin{description}
        \item[In one extreme case: ] The ladder of $B$ boxes starts from the bottom qutrit wire, as shown below. There are $n-1$ $B$ boxes in $L^{(n)}$. According to \cref{fig:Z-and-X-Normal-Boxes}, there are $9$ choices for a $B$ box. Since each choice of a normal box is independent of each other, there are $8 \cdot 3 \cdot 9^{(n-1)} = 24\cdot 9^{(n-1)}$ distinct $L^{(n)}$ when it expands over all $n$ qutrits.
        \[
            \input{figures/NormalForm/ZNormalExtreme.tikz}
        \]
        \item[In the other extreme case: ]The ladder of $B$ boxes starts from the top qutrit wire, as shown below. There is no $B$ box in $L^{(n)}$. Since each choice of a normal box is independent of each other, there are $8 \cdot 3  = 24$ distinct $L^{(n)}$ when it expands over $1$ qutrit.
        \[
            \input{figures/NormalForm/ZNormalExtreme3.tikz}
        \]
        \item[In all other remaining cases: ] The ladder of $B$ boxes starts from qutrit wire $k$, $1 < k < n$. In the illustration below, there are $(k-1)$ $B$ boxes in $L^{(n)}$. Reasoning analogously as before, there are $8 \cdot 3 \cdot 9^{(k-1)} = 24 \cdot 9^{(k-1)}$ distinct $L^{(n)}$ when it expands over $k$ qutrits.
        \[
            \input{figures/NormalForm/ZNormalExtreme2.tikz}
        \]
    \end{description}

    Considering all cases of concatenating B boxes in $L^{(n)}$, the number of distinct $L^{(n)}$ is
    \[
    \sum_{k = 0}^{n-1}(24 \cdot 9^{k}) = 24 \cdot \sum_{k = 0}^{n-1}9^{k} = 3\cdot(9^n -1).
    \]
    \item For $M^{(n)}$, we consider all possible ways of concatenating $D$, $E$, and $F$ boxes. According to \cref{fig:Z-and-X-Normal-Boxes}, there are $3$ choices for $E$ and $F$ boxes respectively. For each $D$ box, there are $9$ choices. According to \cref{fig:X-normal}, there are $n-1$ D boxes in $M^{(n)}$. Since each choice of a normal box is independent of each other, there are $3 \cdot 3 \cdot 9^{(n-1)} = 9 \cdot 9^{(n-1)} = 9^n$ distinct $M^{(n)}$.
\end{itemize}
Since the choice of $L^{(n)}$ and $M^{(n)}$ is independent of each other, there are $(3\cdot(9^n -1)) \cdot 9^n = 3\cdot 9^n (9^n - 1)$ distinct $M^{(n)}L^{(n)}$. By \cref{cor:cardinality}, there are exactly $6$ scalars in $\Clifford_n$. According to the normal form outlined in \cref{fig:macro-normal-circuit}, we can calculate the total number of distinct $N^{(n)}$:
\[
6\cdot \prod_{k=1}^n 3(9^k - 1)9^k.
\]
Since they are in one-to-one correspondence with the elements of the $n$-qutrit Clifford group, this completes the proof.
\end{proof}
\section{A Complete Set of Box Relations}
\label{app:relations}

In \cref{sec:normal-form,sec:normalization-with-relations}, we define a normal form for $\Clifford_n$ and show how to normalize an arbitrary $n$-qutrit Clifford operator. Then, our problem of finding a complete set of Clifford relations reduces to finding a complete set of rewrite rules such that we can push all possible dirty gates through the normal boxes of \cref{fig:Z-and-X-Normal-Boxes}. Since we only need to consider the cases where multi-qutrit (derived) generators acting on adjacent qutrits, each box relation involves at most three qutrits. Therefore, the $2$ zero-qutrit phase relations in \cref{subsec:zero-qutrit-relations}, the $34$ single-qutrit box relations in \cref{subsec:single-qutrit-relations}, the $173$ two-qutrit box relations in \cref{subsec:two-qutrit-relations}, and the $171$ three-qutrit box relations in \cref{subsec:three-qutrit-relations} together form a complete set of box relations for $\Clifford_n$.

\subsection{Zero-Qutrit Box Relations}
\label{subsec:zero-qutrit-relations}
For $n=0$, all Pauli operators are phases. By \cref{cor:phases}, the Clifford operators are phases of the form $(-\omega)^t$, $t \in \Z_6$. Let $-1 = (-\omega)^3$ and $\omega = (-\omega)^4$, then \eqref{eq:scalar-relations} gives the complete set of box relations of $\Clifford_0$.

\begin{equation}
    (-1)^2 = 1, \qquad \omega^3 = 1.
    \label{eq:scalar-relations}
\end{equation}

In \cref{subsec:single-qutrit-relations,subsec:two-qutrit-relations,subsec:three-qutrit-relations}, we express box relations parametrically. This reveals some redundancy in the box relations, allowing for more efficient relation reduction in the supplement~\cite{Supplement2024}. This is also an effective approach for streamlining the derivation of multiple box relations at once.

\subsection{Single-Qutrit Box Relations}
\label{subsec:single-qutrit-relations}

\begin{figure}[H]
    \centering
    \[
    \scalebox{0.9}{\input{figures/BoxRelations/H.A.tikz}}
    \]
    \caption{Pushing $H$ through an $A$ box.}
    \label{fig:H-A-relations}
\end{figure}

\begin{figure}[H]
    \centering
    \[
    \scalebox{0.9}{\input{figures/BoxRelations/S.A.tikz}}
    \]
    \caption{Pushing $S$ through an $A$ box, $b \in \Z_3$.}
    \label{fig:S-A-relations}
\end{figure}

\begin{figure}[H]
    \centering
    \[
    \scalebox{0.9}{\input{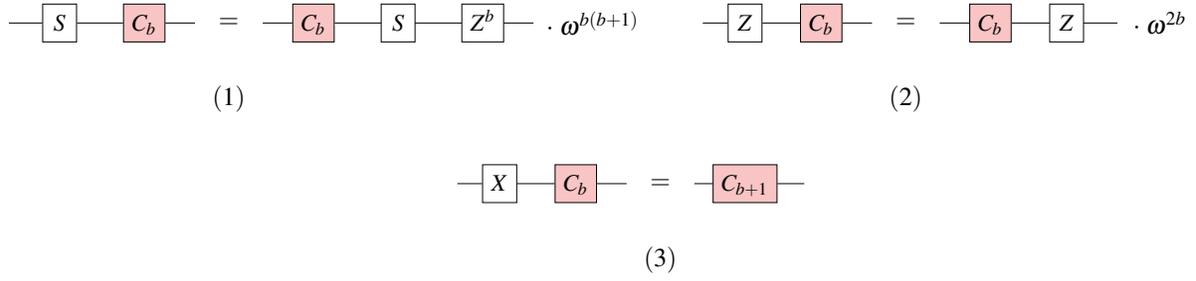}}
    \]
    \caption{Pushing $S$, $Z$, or $X$ through a $C$ box, $b \in \Z_3$.}
    \label{fig:S-Z-X-C-relations}
\end{figure}

\begin{figure}[H]
    \centering
    \[
    \scalebox{0.9}{\input{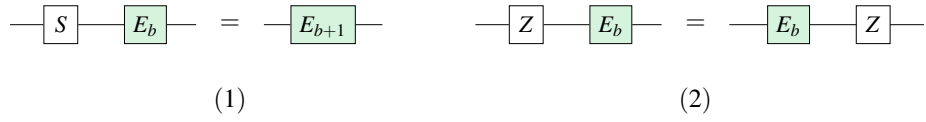}}
    \]
    \caption{Pushing $S$ or $Z$ through an $E$ box, $b\in \Z_3$.}
    \label{fig:S-Z-E-relations}
\end{figure}

\begin{figure}[H]
    \centering
    \[
    \scalebox{0.9}{\input{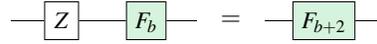}}
    \]
    \caption{Pushing $Z$ through an $F$ box, $b\in \Z_3$.}
    \label{fig:Z-F-relations}
\end{figure}

\subsection{Two-Qutrit Box Relations}
\label{subsec:two-qutrit-relations}

\begin{figure}[H]
\centering
    \[
    \scalebox{0.8}{\input{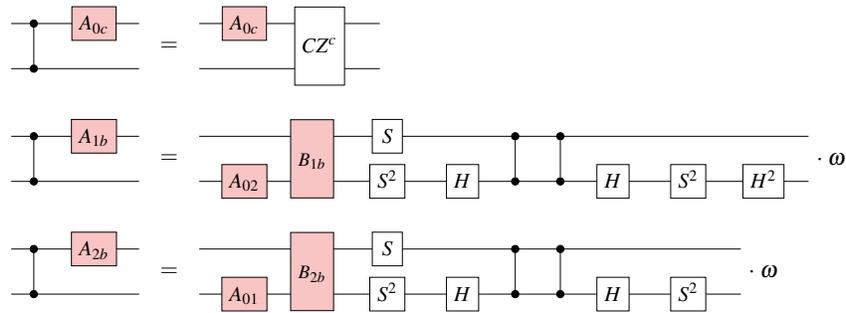}}
    \]
    \caption{Pushing $\CZ$ through an $A$ box, $c \in \{1, 2\}$, $b \in \Z_3$.}
    \label{fig:CZ-A-relations}
\end{figure}

\clearpage
\newgeometry{margin=2cm}

\begin{figure}[H]
    \centering
    \[
    \scalebox{0.8}{\input{figures/BoxRelations/CZ.A01BNew.tikz}}
    \]
    \caption{Pushing $\CZ$ through $A_{01}$ and a $B$ box, $b\in \Z_3$.}
    \label{fig:CZ-A01-B-relations}
\end{figure}

\begin{figure}[H]
    \centering
    \[
    \scalebox{0.8}{\input{figures/BoxRelations/CZ.A02BNew.tikz}}
    \]
    \caption{Pushing $\CZ$ through $A_{02}$ and a $B$ box, $b\in \Z_3$.}
    \label{fig:CZ-A02-B-relations}
\end{figure}

\clearpage
\restoregeometry

\clearpage
\newgeometry{margin=2cm}

\begin{figure}[H]
    \centering
    \[
    \scalebox{0.8}{\input{figures/BoxRelations/CZ.A1bBNew.tikz}}
    \]
    \caption{Pushing $\CZ$ through $A_{1b}$ and a $B$ box, $b,d\in \Z_3$.}
    \label{fig:CZ-A1b-B-relations}
\end{figure}

\begin{figure}[H]
    \centering
    \[
    \scalebox{0.8}{\input{figures/BoxRelations/CZ.A2bBNew.tikz}}
    \]
    \caption{Pushing $\CZ$ through $A_{2b}$ and a $B$ box, $b,d\in \Z_3$.}
    \label{fig:CZ-A2b-B-relations}
\end{figure}

\clearpage
\restoregeometry

\begin{figure}[H]
\centering
    \[
    \scalebox{0.8}{\input{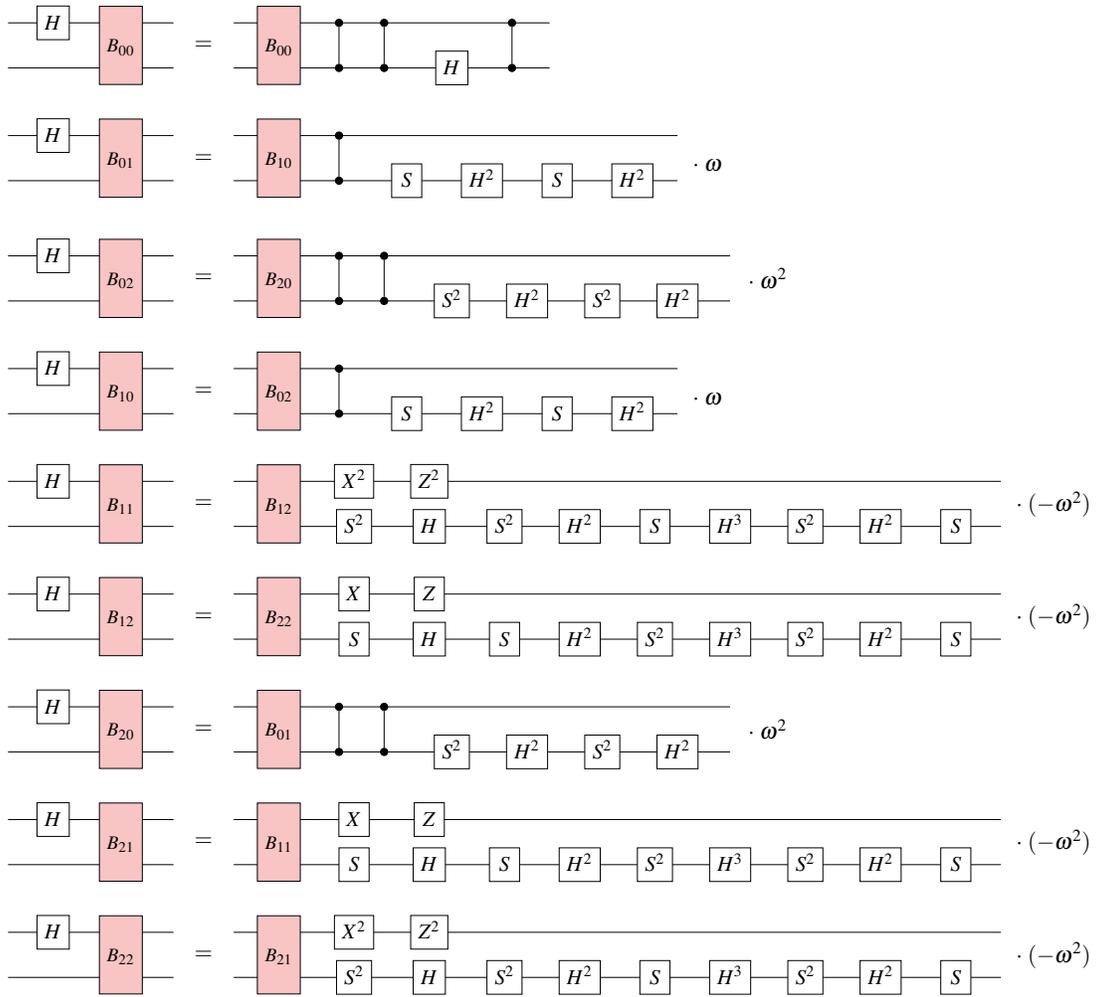}}
    \]
    \caption{Pushing $H\otimes I$ through a $B$ box.}
    \label{fig:HxI-B-relations}
\end{figure}

\begin{figure}[H]
\centering
    \[
    \scalebox{0.8}{\input{figures/BoxRelations/SxI.BNew.tikz}}
    \]
    \caption{Pushing $S \otimes I$ through a $B$ box, $b \in \Z_3$.}
    \label{fig:SxI-B-relations}
\end{figure}

\begin{figure}[H]
\centering
    \[
    \scalebox{0.8}{\input{figures/BoxRelations/IxS.B.tikz}}
    \]
    \caption{Pushing $I \otimes S$ through a $B$ box, $(a,b) \in \Z_3 \times \Z_3 \setminus \{(0,0)\}$.}
    \label{fig:IxS-B-relations}
\end{figure}

\begin{figure}[H]
\centering
    \[
    \scalebox{0.8}{\input{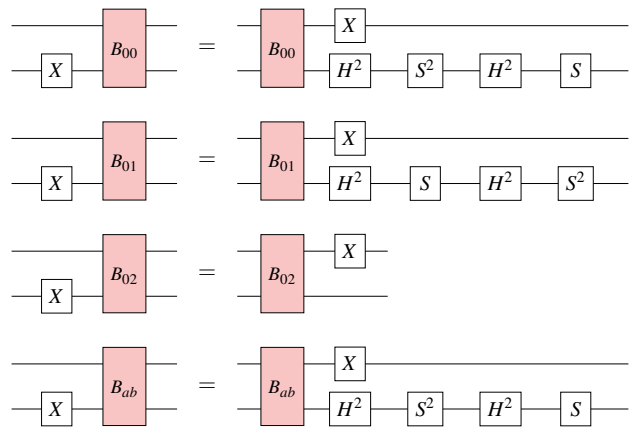}}
    \]
    \caption{Pushing $I \otimes X$ through a $B$ box, $a \in \{1, 2\}$, $b \in \Z_3$.}
    \label{fig:IxX-B-relations}
\end{figure}

\clearpage
\newgeometry{margin=2cm}

\begin{figure}[H]
\centering
    \[
    \scalebox{0.8}{\input{figures/BoxRelations/IxZ.B.tikz}}
    \]
    \caption{Pushing $I \otimes Z$ through a $B$ box, $(a,b) \in \Z_3 \times \Z_3 \setminus \{(0,0)\}$.}
    \label{fig:IxZ-B-relations}
\end{figure}

\begin{figure}[H]
\centering
    \[
    \scalebox{0.8}{\input{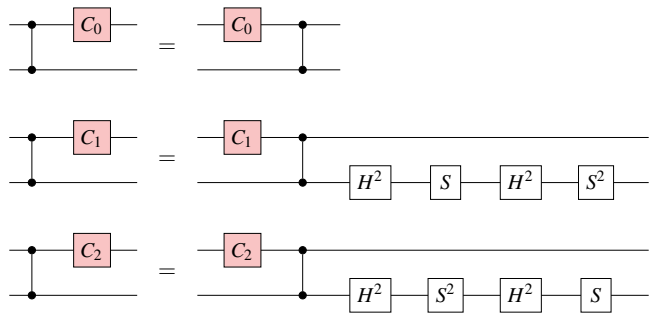}}
    \]
    \caption{Pushing $\CZ$ through a $C$ box.}
    \label{fig:CZ-C-relations}
\end{figure}

\clearpage
\restoregeometry

\begin{figure}[H]
\centering
    \[
    \scalebox{0.8}{\input{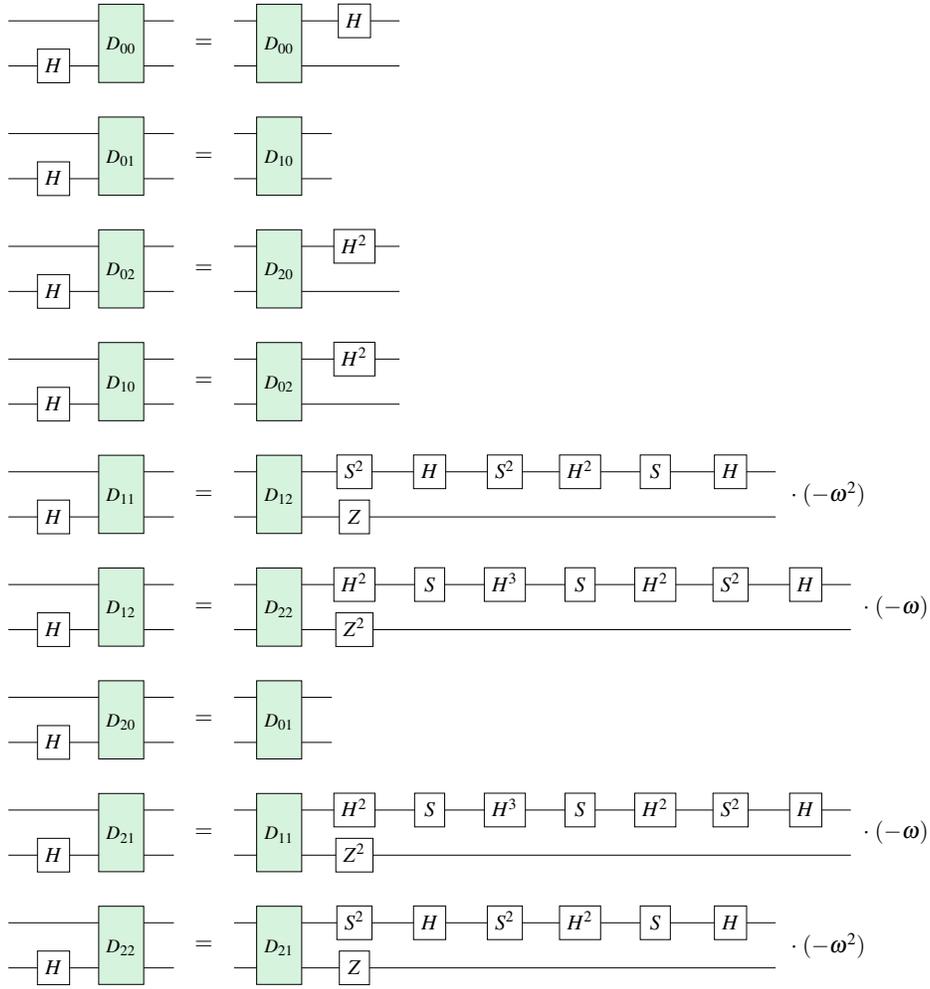}}
    \]
    \caption{Pushing $I \otimes H$ through a $D$ box.}
    \label{fig:IxH-D-relations}
\end{figure}

\begin{figure}[H]
\centering
    \[
    \scalebox{0.8}{\input{figures/BoxRelations/IxS.D.tikz}}
    \]
    \caption{Pushing $I \otimes S$ through a $D$ box, $b \in \Z_3$.}
    \label{fig:IxS-D-relations}
\end{figure}

\clearpage
\newgeometry{margin=2cm}
\begin{figure}[H]
\centering
    \[
    \scalebox{0.8}{\input{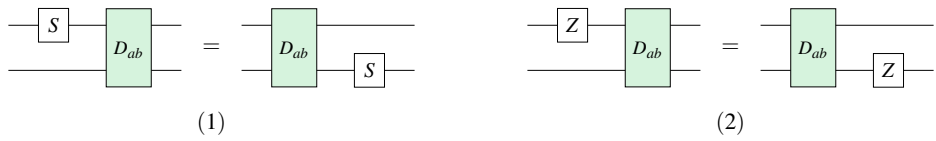}}
    \]
    \caption{Pushing $S \otimes I$ or $Z \otimes I$ through a $D$ box, $a, b \in \Z_3$.}
    \label{fig:SZxI-D-relations}
\end{figure}

\begin{figure}[H]
\centering
    \[
    \scalebox{0.8}{\input{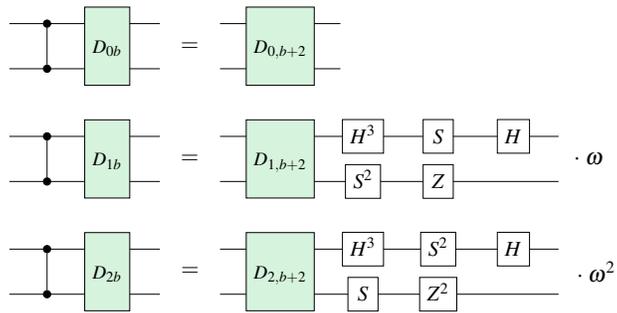}}
    \]
    \caption{Pushing $\CZ$ through a $D$ box, $b \in \Z_3$.}
    \label{fig:CZ-D-relations}
\end{figure}

\subsection{Three-Qutrit Box Relations}
\label{subsec:three-qutrit-relations}

\begin{figure}[H]
\centering
    \[
    \scalebox{0.8}{\input{figures/BoxRelations/CZ.B00.BNew.tikz}}
    \]
    \caption{Pushing $\CZ \otimes I$ through two $B$ boxes—$B_{00}$ followed by another $B$ box, $b \in \Z_3$.}
    \label{fig:CZ-B00-B-relations}
\end{figure}

\clearpage
\restoregeometry

\clearpage
\newgeometry{margin=2cm}

\begin{figure}[H]
\centering
    \[
    \scalebox{0.8}{\input{figures/BoxRelations/CZ.B01.BNew.tikz}}
    \]
    \caption{Pushing $\CZ \otimes I$ through two $B$ boxes—$B_{01}$ followed by another $B$ box, $b \in \Z_3$.}
    \label{fig:CZ-B01-B-relations}
\end{figure}

\begin{figure}[H]
\centering
    \[
    \scalebox{0.8}{\input{figures/BoxRelations/CZ.B02.BNew.tikz}}
    \]
    \caption{Pushing $\CZ \otimes I$ through two $B$ boxes—$B_{02}$ followed by another $B$ box, $b \in \Z_3$.}
    \label{fig:CZ-B02-B-relations}
\end{figure}

\clearpage
\restoregeometry

\clearpage
\newgeometry{margin=2cm}

\begin{figure}[H]
\centering
    \[
    \scalebox{0.8}{\input{figures/BoxRelations/CZ.B1b.BNew.tikz}}
    \]
    \caption{Pushing $\CZ \otimes I$ through two $B$ boxes—$B_{1b}$ followed by another $B$ box, $b,d \in \Z_3$.}
    \label{fig:CZ-B1b-B-relations}
\end{figure}

\begin{figure}[H]
\centering
    \[
    \scalebox{0.8}{\input{figures/BoxRelations/CZ.B2b.BNew.tikz}}
    \]
    \caption{Pushing $\CZ \otimes I$ through two $B$ boxes—$B_{2b}$ followed by another $B$ box, $b,d \in \Z_3$.}
    \label{fig:CZ-B2b-B-relations}
\end{figure}

\clearpage
\restoregeometry

\clearpage
\newgeometry{margin=2cm}

\begin{figure}[H]
\centering
    \[
    \scalebox{0.8}{\input{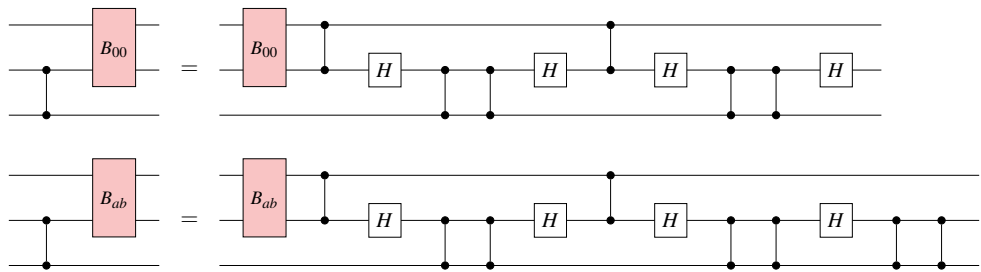}}
    \]
    \vspace{-.3 cm}
    \caption{Pushing $\CZ$ through a $B$ box, $(a,b) \in \Z_3 \times \Z_3 \setminus \{(0,0)\}$.}
    \label{fig:CZ-B-relations}
\end{figure}

\vspace{-.7 cm}

\begin{figure}[H]
\centering
    \[
    \scalebox{0.8}{\input{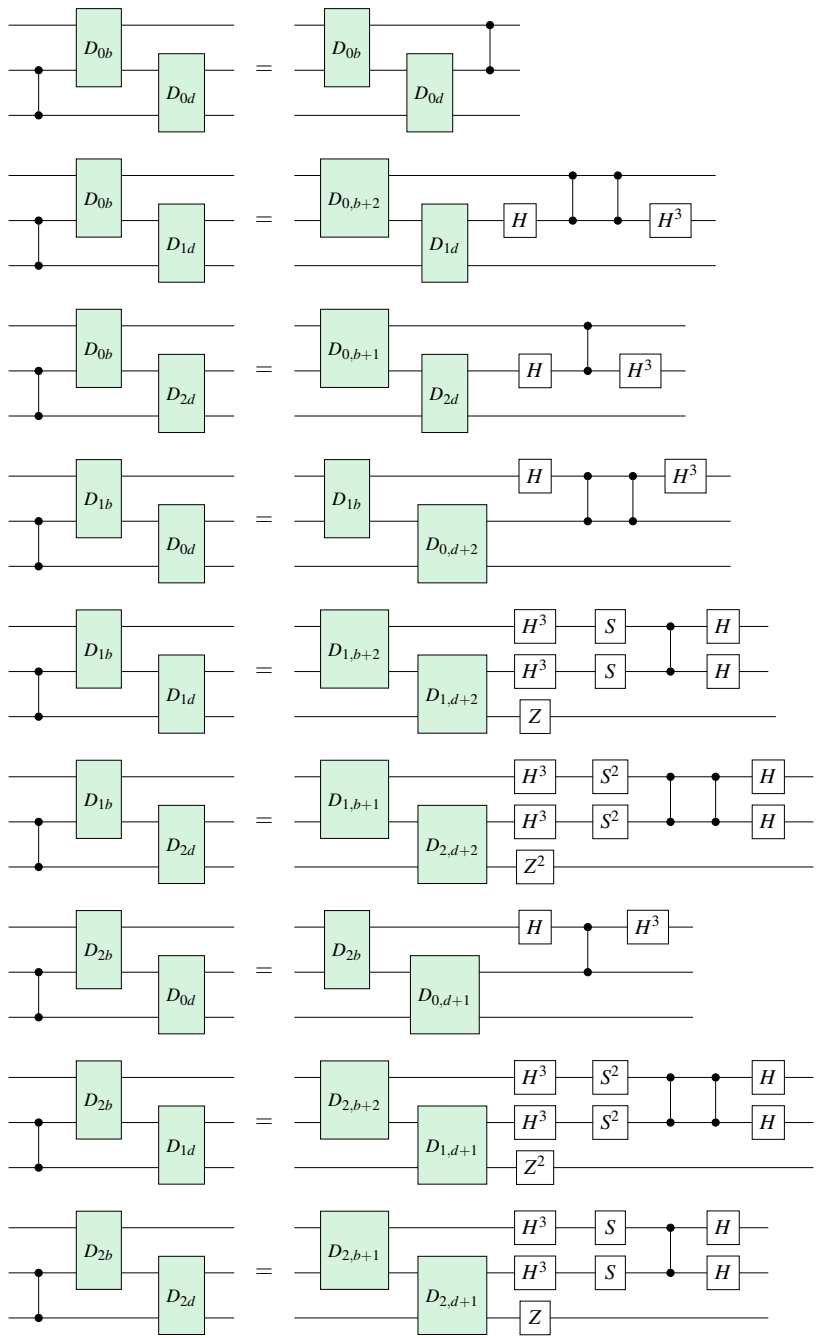}}
    \]
    \vspace{-.3 cm}
    \caption{Pushing $I \otimes \CZ$ through two $D$ boxes, $b,d \in \Z_3$.}
    \label{fig:CZ-D-D-relations}
\end{figure}

\clearpage
\restoregeometry
\section{Actions of the Clifford Generators}
\label{app:generator-actions}

According to \cref{subsec:action}, we can describe a (derived) Clifford generator by its conjugation on the Pauli generators. That is, let $C\in\Clifford_n$ and $P\in\mathcal{B}_n$, $CPC^\dagger = Q$, for some $Q\in \Pauli_n$. For our purposes, it is also convenient to describe their inverse conjugation on the Pauli generators: $C^\dagger PC = Q'$, for some $Q' \in \Pauli_n$.

\begin{definition}
    $H$, $S$, $X$, and $Z$ can be described by their actions on $X$ and $Z$.
    \[
    \input{figures/Preliminaries/SingleQutritAuto.tikz}
    \]
    \label{def:Single-qutrit-action}
\end{definition}

\begin{definition}
    $CZ$ and $SWAP$ can be described by their actions on $X\otimes I$, $Z\otimes I$, $I \otimes X$, and $I \otimes Z$.
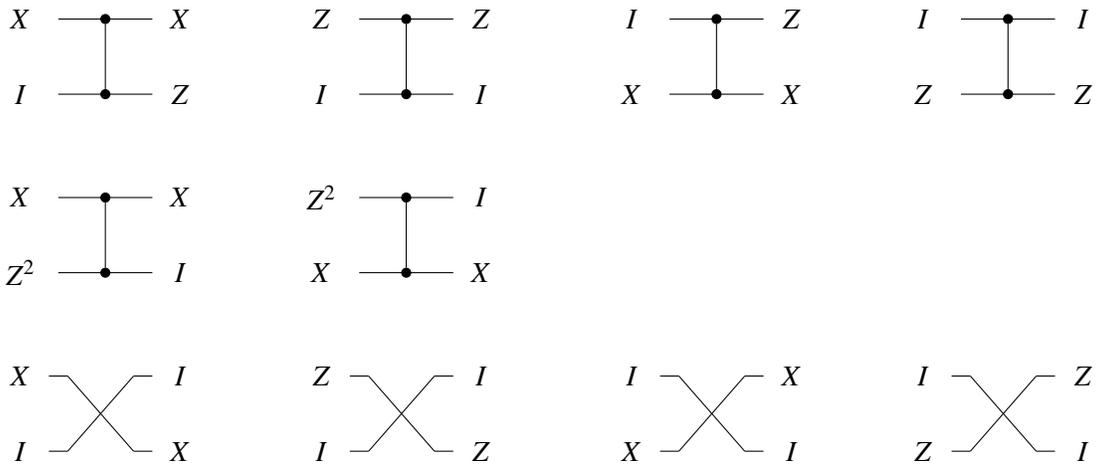
    \[
    \input{figures/Preliminaries/TwoQutritAuto.tikz}
    \]
    \label{def:Two-qutrit-action}
\end{definition}

\newpage
\section{Designing the Normal Boxes}
\label{sec:design}

The multi-qutrit Clifford normal form introduced in \cref{sec:normal-form} generalizes the normal forms designed for qubit Clifford operators~\cite{makary2021generators,selinger2015generators}. We adopt a similar methodology and leverage the underlying mechanics of the normal form to handle the increased degrees of freedom for qutrits. Note that the choice of normal boxes is neither unique nor arbitrary. Some choices are much nicer for the Clifford normalization, because they produce simpler residual dirty gates. In this section, we outline the key ideas behind the design of these normal boxes, paving the way for generalizing multi-qutrit Clifford completeness to other odd prime dimensions.

Recall that $\Clifford_n$ is the collection of $n$-qutrit Clifford operators generated by $-\omega$, $H$, $S$, and $\CZ$ gates through matrix multiplication and tensor product. Let $\mathcal{A}=\{-\omega,H^{(j)},S^{(j)},\CZ^{(i,i+1)};\; 1\leq j\leq n,\; 1\leq i\leq n-1\}$, where $-\omega$, $H^{(j)}$, $S^{(j)}$, and $\CZ^{(i,i+1)}$ denote the $n$-qutrit unitary operators:

\begin{align*}
    -\omega &= (-\omega)\cdot I_{3^n},\quad I_{3^n} \; \mbox{is the}\; {3^n}\times {3^n} \; \mbox{identity matrix}.\\
    H^{(j)} &= \bigotimes_{\ell=1}^n C_\ell,\quad C_\ell = H \mbox{ if } \ell = j, \quad  C_\ell = I_3 \; \mbox{otherwise}.\\
    S^{(j)} &= \bigotimes_{\ell=1}^n C_\ell,\quad C_\ell = S \mbox{ if } \ell = j, \quad C_\ell = I_3 \; \mbox{otherwise}.\\
    CZ^{(i,i+1)} &= I_{3^{i-1}} \otimes \CZ \otimes I_{3^{n-i-1}}.
\end{align*}

Every element of $\mathcal{A}$ is called a \emph{syllable}. We use $\mathcal{A}^*$ to denote the set of all words formed by the syllables in $\mathcal{A}$. These words are composed in diagrammatic order. For $1\leq k\leq m$, if $\mathbf{W}=A_1\ldots A_m$ with $A_k \in \mathcal{A}$, we call $\mathbf{W}$ a \emph{Clifford word over $\mathcal{A}$}, and note that $\mathbf{W} \in \mathcal{A}^*$. Let $\llbracket A_k\rrbracket$ denote the matrix representation of $A_k$. Then the word $\mathbf{W}$ can be \emph{interpreted} as a unitary in $\Clifford_n$ by multiplying $\llbracket A_k\rrbracket$ in the reverse order of the syllable composition,

\[
\llbracket \mathbf{W}\rrbracket= \left\llbracket A_1\ldots A_m\right\rrbracket=\llbracket A_m\rrbracket\cdots \llbracket A_1\rrbracket.
\]

To normalize $\mathbf{W}$, we first append it on the right with the normal form of the identity operator, as given in \eqref{eq:n-qutrit-identity-normal}. \cref{fig:dirty-normal-form,fig:normalization1} show that Clifford normalization is reduced to one of the two cases: pushing an $H$ or $S$ gate into a Clifford normal form; or pushing a $\CZ$ gate into a Clifford normal form. Locally, it suffices to know how to push a Clifford gate through each of the normal boxes specified in \cref{fig:Z-and-X-Normal-Boxes}.


\begin{figure}[!htb]
    \centering
    \[
 \scalebox{0.8}{\input{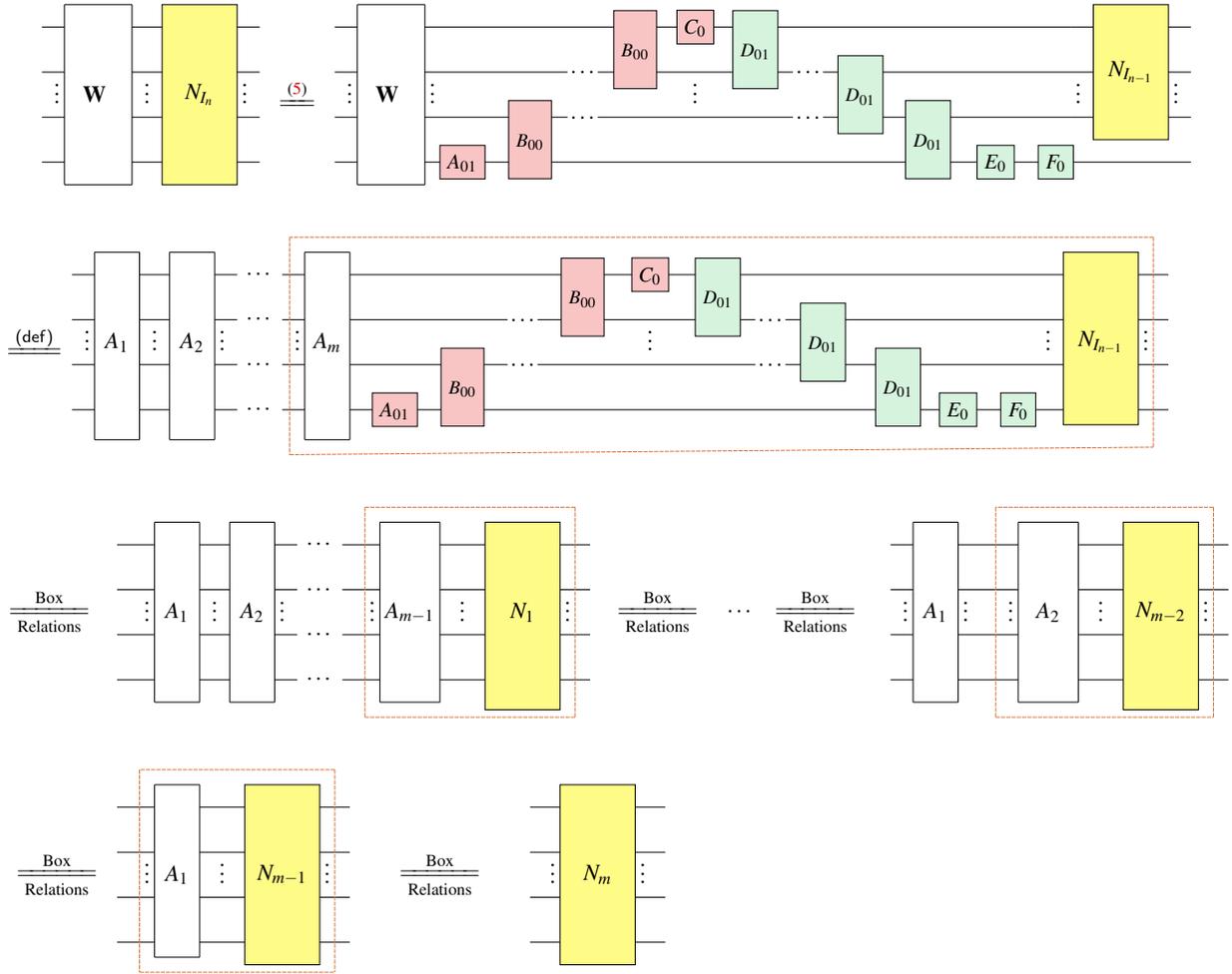}}
\]
    \caption{$\mathbf{W}=A_1\ldots A_m$ is a Clifford word composed of $m$ syllables. Each syllable is either a single-qutrit $H$ or $S$ gate, or a two-qutrit $\CZ$ gate acting on adjacent qutrits. After pushing $A_m$ into the normal form of the $n$-qutrit identity operator, $N_{I_n}$, this normal form is updated to a new normal form $N_1$. Then, $A_{m-1}$ is pushed into $N_1$ and so on and so forth. The normalization process reduces to pushing each Clifford syllable $A_k$ into the updated Clifford normal form.}
    \label{fig:normalization1}
\end{figure}

Based on \cref{fig:micro-normal-circuit,fig:Z-and-X-Normal-Boxes}, \cref{fig:dirty-gate-restrictions-1,fig:dirty-gate-restrictions-2,fig:dirty-gate-restrictions-3} summarize the `allowed' Clifford operators before $B$, $C$, $D$, $E$, and $F$ boxes. This, on the other hand, introduces the dirty gate restrictions when pushing through a normal box. To ensure the Clifford normalization will eventually terminate, we need to account for these restrictions when designing the normal boxes.

As discussed in \cref{sec:normalization-with-relations}, after pushing a Clifford operator $g$ through a normal box $M$, the updated normal box $M'$ is uniquely determined by $M$ and $g$. This is due to the unique box actions specified in the second column of \cref{fig:normal-box-auto}. $M'$ is followed by a sequence of Clifford gates, referred to as the residual dirty gates, $\dir$, and it is determined by $g$, $M$, and $M'$ collectively. More generally, when $M$ is a certain combination of normal boxes, \cref{proc:find-d} provides a recipe for finding $\dir$ by tracking the Clifford conjugation on the Pauli generators.

\begin{procedure}
Let $n \in \N$ and $g \in \Clifford_n$. Let $M$ be some combination of normal boxes and $M'$ be the updated normal boxes after pushing $g$ through $M$. Based on the dirty normal form illustrated in \cref{fig:dirty-normal-form}, $n$ can be either $1$, $2$, or $3$, depending on the choice of $M$. We can find the residual dirty gates $\dir$ as follows. 

\begin{equation}
    \input{figures/Generalizability/finddnew.tikz}
    \label{eq:find-d}
\end{equation}
    
\begin{enumerate}
    \item Let $T \in \{X^{(j)},Z^{(j)};\;1\leq j\leq n\}$. Find the preimage of $T$ in the lefthand side of \eqref{eq:find-d}. Let it be $P$.
    \[
    \input{figures/Generalizability/findd1new.tikz}
    \]
    \item Find the image of $P$ in $M'$. Let it be $Q$.
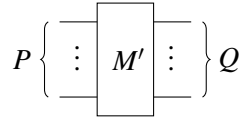
    \[
    \input{figures/Generalizability/findd2new.tikz}
    \]
    \item Then the automorphism of $\dir$ is given by $Q$ and $T$.
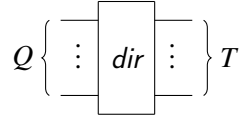
    \[
    \input{figures/Generalizability/findd3new.tikz}
    \]
\end{enumerate}
    \label{proc:find-d}
\end{procedure}

Next, we discuss the design principles behind the specific implementation of the normal boxes in \cref{fig:Z-and-X-Normal-Boxes}. In \cref{subsec:normalize-single}, we consider the case of normalizing a single-qutrit Clifford operator, which boils down to pushing an $H$ or $S$ gate through the $A$, $C$, $E$, and $F$ boxes. In \cref{subsec:normalize-multi}, we consider the case of normalizing a multi-qutrit Clifford operator, which boils down to pushing an $H$, $S$ or $\CZ$ gate through the $A$, $B$, $C$, and $D$ boxes. In both cases, our choice of normal boxes satisfy the dirty gate restrictions outlined in \cref{fig:dirty-gate-restrictions-1,fig:dirty-gate-restrictions-2,fig:dirty-gate-restrictions-3}.

\subsection{Designing the Single-Qutrit Normal Boxes}
\label{subsec:normalize-single}

In \cref{sec:normal-form}, we introduce four types of single-qutrit normal boxes. \cref{fig:single-qutrit-normal-boxes-auto} summarizes their unique actions on Pauli operators. These properties completely determine all of the boxes, except for the $A$ boxes.
In particular, since a $C_b$ box maps $X$ to $X$ and $\omega^b Z$ to $Z$, it must be a power of $X$ gates. Since an $E_b$ box maps $XZ^b$ to $X$ and $Z$ to $Z$, it must be a power of $S$ gates. For all $i\in \Z_3$, since an $F_b$ box maps $\omega^bXZ^i$ to $XZ^i$ and $Z$ to $Z$, it maps $\omega^bX$ to $X$. Hence, an $F_b$ box must be a power of $Z$ gates. By \cref{def:Single-qutrit-action}, $C_b = X^b$, $E_b = S^b$, and $F_b = Z^{2b}$, which gives us the construction of $C$, $E$, and $F$ boxes in \cref{fig:Z-and-X-Normal-Boxes}. 

\begin{figure}[!htb]
    \centering
    \[
    \input{figures/Generalizability/single-qutritNormalBoxesAuto.tikz}
\]
    \caption{Single-qutrit normal boxes and their unique actions on Pauli operators. The second column shows the required actions. The third column shows the additional actions that we choose to simplify the normalization process.}
    \label{fig:single-qutrit-normal-boxes-auto}
\end{figure}

Based on this, \cref{fig:single-qutrit-normal} presents a generic construction of a single-qutrit Clifford normal form $N$, and \cref{fig:single-qutrit-restriction} summarizes the dirty gate restrictions when pushing an $H$ or $S$ gate through $N$.

\begin{figure}[!htb]
    \centering
    \[
\input{figures/Generalizability/Single-qutritNormal2.tikz}
\]
    \caption{For $(a,b) \in \Z_3 \times \Z_3 \setminus \{(0,0)\}$, $c,\nu, u \in \Z_3$, and $t\in \Z_6$, a generic construction of the normal form of a single-qutrit Clifford operator. This is based on the close-up view of the Clifford normal form in \cref{fig:micro-normal-circuit} and the discussion above.}
    \label{fig:single-qutrit-normal}
\end{figure}

\begin{figure}[!htb]
    \centering
    \[
\input{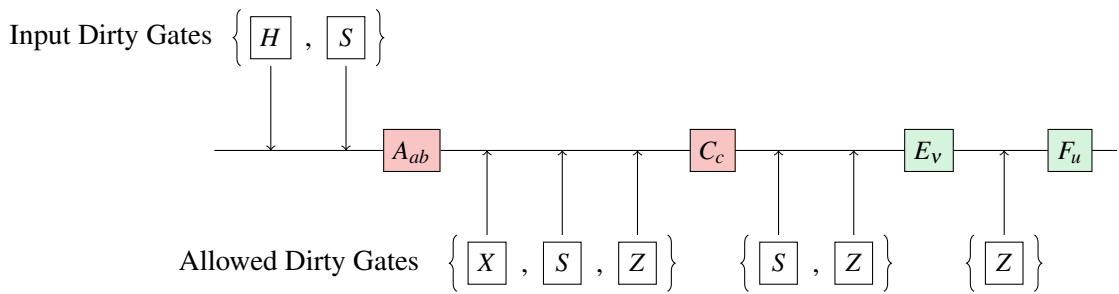}
\]
    \caption{Dirty gate restrictions when pushing an $H$ or $S$ gate through a single-qutrit normal form.}
    \label{fig:single-qutrit-restriction}
\end{figure}

We thus see that $A$ boxes must be able to handle arbitrary dirty gates at the input, and the residual dirty gates must not contain solitary $H$ gates—these can only be absorbed by the subsequent normal boxes if they appear as part of a Pauli $X$ or $Z$ gate. Therefore, to find $A$ boxes that are suitable for normalization, we need to consider the case when pushing an $H$ or $S$ gate through $A_{ab}$. The updated $A$ box is uniquely determined by the tuple $(a, b, g)$, $g \in \{H, S\}$.

\begin{lemma}
    For all $(a,b)\in\Z_3\times \Z_3\setminus\{(0,0)\}$, there exist $\dir,\;\dir' \in \Clifford_1$ such that
    \[
    \input{figures/NormalForm/HSA.tikz}
    \] 
    \label{lem:HA}
\end{lemma}

\begin{proof}
    By the uniqueness of normal form and the required action of $A$ boxes in \cref{fig:single-qutrit-normal-boxes-auto}, it suffices to find the preimage of $X^aZ^b$ under the action of $H$ and $S$ gates.

    \[
        H^\dagger (X^aZ^b) H = (H^\dagger X^a H)(H^\dagger Z^b H)\xlongequal[]{\cref{def:Single-qutrit-action}}Z^{-a}X^b\xlongequal[]{\cref{def:Single-qutrit-action}}\omega^{-ab}X^bZ^{-a}.
    \]

    Hence, after pushing $H$ gate through an $A_{ab}$ box, the updated $A$ box must be of the form $A_{b,-a}$.
    \[
    \input{figures/Generalizability/H.A.tikz}
    \]

    \begin{align*}
         S^\dagger (X^aZ^b) S = (S^\dagger X^a S)(S^\dagger Z^b S)&\xlongequal[]{\cref{def:Single-qutrit-action}}\left((XZ^{-1})\cdots(XZ^{-1})\right)Z^b\\
         &\xlongequal[]{\cref{def:Single-qutrit-action}}\omega^{\frac{-a(a-1)}{2}}X^aZ^{-a}Z^b=\omega^{\frac{-a(a-1)}{2}}X^aZ^{b-a}.
    \end{align*}
    Hence, after pushing an $S$ gate through an $A_{ab}$ box, the updated $A$ box must be of the form $A_{a,b-a}$.
    \[
        \scalebox{.9}{\input{figures/Generalizability/S.A.tikz}}\qedhere
    \] 
\end{proof}

Combining \cref{lem:HA,proc:find-d} as well as the choice of $A$ boxes in \cref{fig:Z-and-X-Normal-Boxes}, we can find all single-qutrit $A$ box relations. They are demonstrated in \cref{fig:H-A-relations,fig:S-A-relations}. The dirty gates after pushing an $H$ or $S$ through an $A$ box satisfy the restrictions outlined in \cref{fig:single-qutrit-restriction}.

Putting everything together, \cref{fig:normalize single-qutrit Clifford} demonstrates how to normalize a single-qutrit Clifford operator $\mathbf{W}$, as shown below. We begin by appending $\mathbf{W}$ on the right with the normal form of the identity operator given in \eqref{eq:identities}. Then we iteratively push every Clifford gate through the normal form, by applying the rules in \cref{subsec:single-qutrit-relations}.

\vspace{-.5 cm}

\[
        \scalebox{.9}{\input{figures/Generalizability/example.tikz}}
    \]

\begin{figure}[!htb]
    \centering
    \[
        \scalebox{.9}{\input{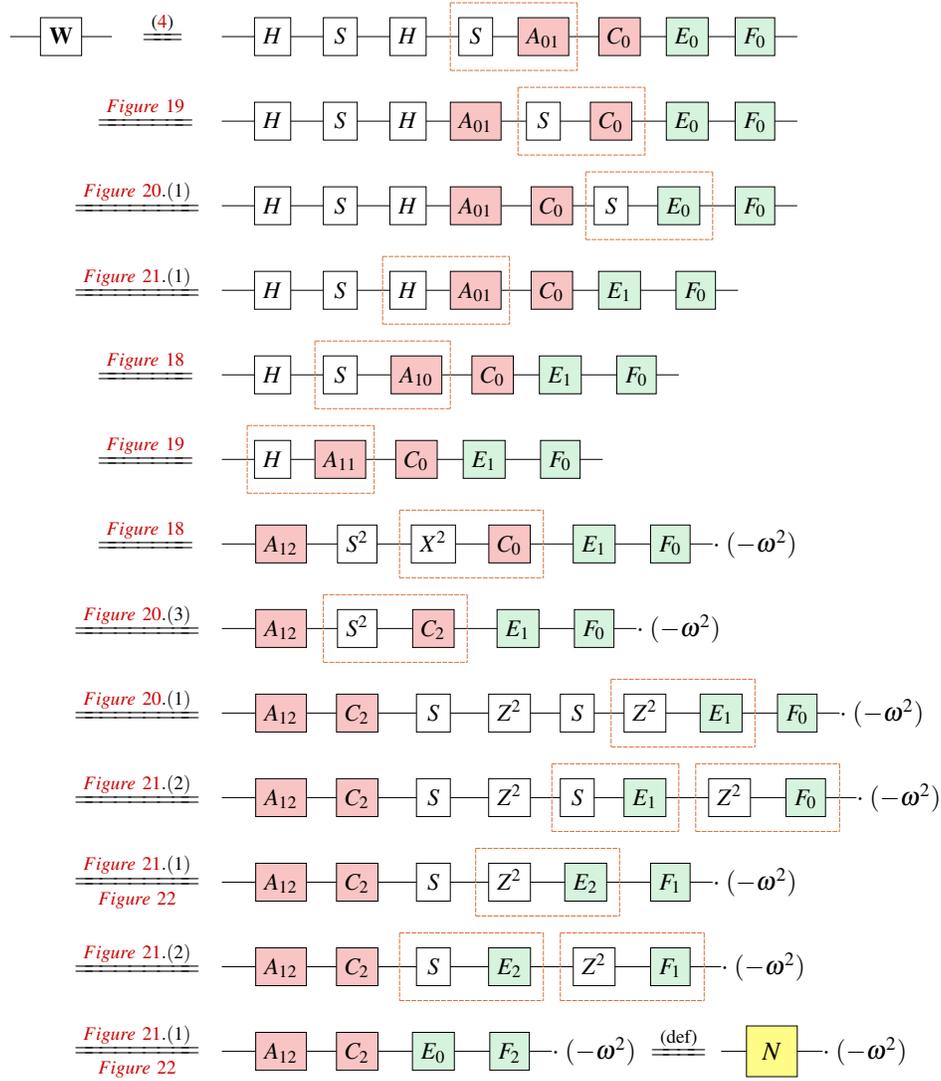}}
    \]
    \caption{Normalize a single-qutrit Clifford operator. Each orange dashed box highlights the subcircuit where a box relation is applied.}
    \label{fig:normalize single-qutrit Clifford}
\end{figure}

\subsection{Designing the Two-Qutrit Normal Boxes}
\label{subsec:normalize-multi}

It is much more involved in finding the appropriate constraints for designing the two-qutrit normal boxes. In \cref{fig:normal-box-auto}, we introduce two types of such boxes, whose distinct actions on Pauli operators are summarized in \cref{fig:two-qutrit-normal-boxes-auto}. Both actions defined for the $D$ box are essential to ensure the normal form behaves correctly. For the $B$ box, however, the second action is identified as a desirable property, as we will explain later. Note that, in all cases, we specify at most two Pauli actions per box. As a result, the stabilizer tableaus are underspecified, leaving a large design space for the normal boxes. To make it more tractable, we introduce a few guiding principles.

\begin{figure}[!htb]
    \centering
    \[
    \input{figures/Generalizability/two-qutritNormalBoxesAuto.tikz}
\]

    \caption{Two-qutrit normal boxes and their unique actions on Pauli operators. The second column shows the required actions. The third column shows the additional actions that we choose to simplify the normalization process.}
    \label{fig:two-qutrit-normal-boxes-auto}
\end{figure}

Recall that in \cref{fig:single-qutrit-normal-boxes-auto}, an $A_{ab}$ box sends $X^aZ^b$ to $Z$. This is the same as the Pauli evolution happening on the top input wire of $B$ boxes. To not reinvent the wheel, it is desirable to construct $B$ boxes using $A$ boxes. Moreover, the required action of $B$ and $D$ boxes on Pauli operators is almost vertically symmetric of each other. Ideally, we would like to design $D$ boxes based on how we design $B$ boxes. For example, we also want to use $A$ boxes to construct $D$ boxes.

Finally, among all possible box relations, \cref{fig:bottle-neck} shows the most complicated case where a $\CZ$ is pushed through two consecutive $B$ or $D$ boxes. In each case, there are $81$ box relations. It is desirable to choose appropriate $B$ and $D$ boxes such that the dirty gates in each one of the $162$ box relations are as simple as possible.

\begin{figure}[!htb]
    \centering
    \[
    \input{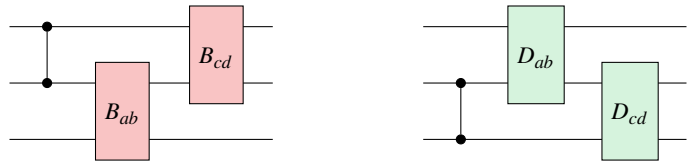}
\]
    \caption{For all $a,b,c,d\in\Z_3$, push a $\CZ$ through two consecutive B or D boxes.}
    \label{fig:bottle-neck}
\end{figure}

As before, we can uniquely determine the updated $B$ and $D$ boxes after pushing a $\CZ$ through them.

\begin{figure}[!htb]
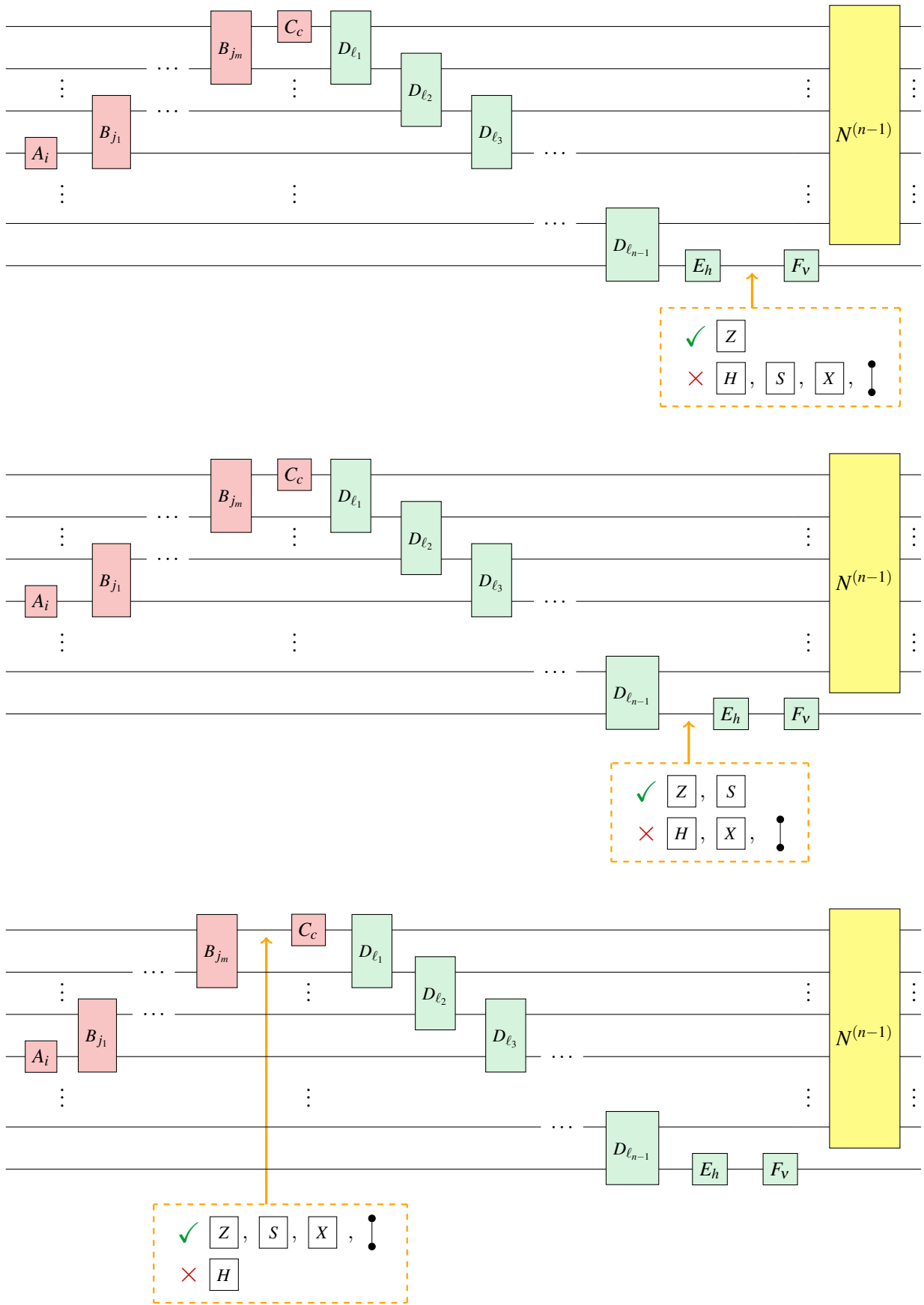

\begin{subfigure}{1\textwidth}
  \centering
   \scalebox{1}{\input{figures/Generalizability/push1.tikz}}
\end{subfigure}%

\vspace{2 em}

\begin{subfigure}{1\textwidth}
  \centering
    \scalebox{1}{\input{figures/Generalizability/push2.tikz}}
\end{subfigure}%

\vspace{2 em}

\begin{subfigure}{1\textwidth}
  \centering
    \scalebox{1}{\input{figures/Generalizability/push4.tikz}}
\end{subfigure}%
\caption{Allowed Clifford operators before $C$, $E$, and $F$ boxes.}
\label{fig:dirty-gate-restrictions-1}
\end{figure}

\begin{figure}[!htb]
\begin{subfigure}{1\textwidth}
  \centering
   \scalebox{1}{\input{figures/Generalizability/push5.tikz}}
\end{subfigure}

\vspace{4 em}

\begin{subfigure}{1\textwidth}
  \centering
    \scalebox{1}{\input{figures/Generalizability/push3.tikz}}
\end{subfigure}
\caption{Allowed Clifford operators before $D$ boxes.}
\label{fig:dirty-gate-restrictions-2}
\end{figure}

\begin{figure}[!htb]
\begin{subfigure}{1\textwidth}
  \centering
   \scalebox{1}{\input{figures/Generalizability/push6New.tikz}}
\end{subfigure}

\vspace{4 em}

\begin{subfigure}{1\textwidth}
  \centering
    \scalebox{1}{\input{figures/Generalizability/push7.tikz}}
\end{subfigure}
\caption{Allowed Clifford operators before $B$ boxes.}
\label{fig:dirty-gate-restrictions-3}
\end{figure}

\begin{lemma}
    For $a, b, c, d \in \Z_3$ and $\dir \in \Clifford_3$,
    \begin{align}
        \scalebox{1}{\input{figures/Generalizability/CZ.B.BUpdated.tikz}}\label{eq:BB}\\
        \scalebox{1}{\input{figures/Generalizability/CZ.D.DUpdated.tikz}}\label{eq:DD}
    \end{align}
    \label{lem:CZ.B.B-CZ.D.D-updated}
\end{lemma}
\begin{proof}
To determine the updated indices of the $B$ boxes in the righthand side of \eqref{eq:BB}, it suffices to find the preimage of $Z\otimes I \otimes I$ in the lefthand side of \eqref{eq:BB}. By the uniqueness of normal form and the required action of $B$ boxes in \cref{fig:two-qutrit-normal-boxes-auto}, our problem is reduced to finding the preimage of $X^cZ^d \otimes X^aZ^b$ in the $\CZ$ circuit.
\begin{align*}
    \CZ^\dagger(X^cZ^d \otimes X^aZ^b)CZ &= \CZ^\dagger\left((X^c \otimes I)(Z^d \otimes I)(I \otimes X^a)(I \otimes Z^b)\right)\CZ\\
    &=\left(\CZ^\dagger(X^c \otimes I)\CZ\right)\left(\CZ^\dagger(Z^d \otimes I)\CZ\right)\left(\CZ^\dagger(I \otimes X^a)\CZ\right)\left(\CZ^\dagger(I \otimes Z^b)\CZ\right)\\
    &=\left(\CZ^\dagger(X \otimes I)\CZ\right)^c\left(\CZ^\dagger(Z \otimes I)\CZ\right)^d\left(\CZ^\dagger(I \otimes X)\CZ\right)^a\left(\CZ^\dagger(I \otimes Z)\CZ\right)^b\\
    &\xlongequal[]{\cref{def:Two-qutrit-action}}(X \otimes Z^2)^c(Z \otimes I)^d(Z^2 \otimes X)^a(I \otimes Z)^b\\
    &=(X^c \otimes Z^{2c})(Z^d \otimes I)(Z^{2a} \otimes X^a)(I \otimes Z^b)\\
    &=X^cZ^dZ^{2a} \otimes Z^{2c} X^aZ^b\xlongequal[]{\cref{def:Single-qutrit-action}}\omega^{2ac}X^cZ^{d+2a} \otimes X^aZ^{b+2c}.
\end{align*}
   \[
    \scalebox{.8}{\input{figures/Generalizability/CZ.B.B1.tikz}}
    \]
    To determine the updated indices of the $D$ boxes in the righthand side of \eqref{eq:DD}, it suffices to find the preimage of $I\otimes I \otimes XZ^i$, for $i\in \Z_3$, in the lefthand side of \eqref{eq:DD}. By the uniqueness of normal form and the required action of $D$ boxes in \cref{fig:two-qutrit-normal-boxes-auto}, our problem is reduced to finding the preimage of $X^aZ^b \otimes X^cZ^d$ in the $\CZ$ circuit.
    
    \begin{align*}
    \CZ^\dagger(X^aZ^b\otimes X^cZ^d)CZ &= \CZ^\dagger\left((X^a \otimes I)(Z^b \otimes I)(I \otimes X^c)(I \otimes Z^d)\right)\CZ\\
    &=\left(\CZ^\dagger(X^a \otimes I)\CZ\right)\left(\CZ^\dagger(Z^b \otimes I)\CZ\right)\left(\CZ^\dagger(I \otimes X^c)\CZ\right)\left(\CZ^\dagger(I \otimes Z^d)\CZ\right)\\
    &=\left(\CZ^\dagger(X \otimes I)\CZ\right)^a\left(\CZ^\dagger(Z \otimes I)\CZ\right)^b\left(\CZ^\dagger(I \otimes X)\CZ\right)^c\left(\CZ^\dagger(I \otimes Z)\CZ\right)^d\\
    &\xlongequal[]{\cref{def:Two-qutrit-action}}(X \otimes Z^2)^a(Z \otimes I)^b(Z^2 \otimes X)^c(I \otimes Z)^d\\
    &=(X^a \otimes Z^{2a})(Z^b \otimes I)(Z^{2c} \otimes X^c)(I \otimes Z^d)\\
    &=X^aZ^bZ^{2c} \otimes Z^{2a} X^cZ^d\xlongequal[]{\cref{def:Single-qutrit-action}}\omega^{2ac}X^aZ^{b+2c}\otimes X^cZ^{d+2a}.
\end{align*}
    \[
    \scalebox{.8}{\input{figures/Generalizability/CZ.D.D1.tikz}}
    \]
\end{proof}

According to the dirty gate restrictions in \cref{fig:dirty-gate-restrictions-2,fig:dirty-gate-restrictions-3}, our proposed $B$ box construction is suitable for Clifford normalization if there is no single $H$ gate acting on the top wire of $\dir$. One way to do this is to show that $\dir$ is not entangling the first qutrit with the other qutrits, and that only Pauli gates appear on the first qutrit. To this end, let us consider the notion of \emph{separable operators}. We will use their properties to decide if the normal boxes are designed appropriately.

\begin{definition}
Let $k_1,k_2\in \N^{>0}$. When a Clifford operator $C \in \Clifford_{k_1 + k_2}$ is a tensor product of two Clifford operators acting on the first $k_1$ and the remaining $k_2$ qutrits respectively, we say $C$ is \emph{separable}.
    \label{def:separable}
\end{definition}

When $C$ is separable, it cannot propagate Paulis from the first $k_1$ qutrits to the other $k_2$ qutrits. In \cref{prop:untangled}, we show that the converse is true. 

\begin{definition}\label{def:commutant}
    For any subset $\mathcal{A}\subset \mathcal{M}_{n}(\C)$, its commutant is defined as
\begin{equation*}
   \mathcal{A}'=\{B\in \mathcal{M}_{n}(\C);\; [B,A]=0 \text{ for all }A\in \mathcal{A} \}.
\end{equation*}
\end{definition}

Note that the commutation of all matrices is just the scalars: $\mathcal{M}_{n}(\C)'=\{cI_{n}:c\in \C\}$. In a composite system $\mathcal{M}_{n_1}(\C)\otimes \mathcal{M}_{n_2}(\C)$,
\begin{equation*}
    \big(\mathcal{M}_{n_1}(\C)\otimes I_{n_2} \big)'=I_{n_1}\otimes \mathcal{M}_{n_2}(\C).
\end{equation*}
Here $\mathcal{M}_{n_1}(\C)\otimes I_{n_2}$ denotes the set $\{A\otimes I_{n_2}:A\in \mathcal{M}_{n_1}(\C)\}$. The following lemma is well known; we include a proof to keep this appendix self-contained.

\begin{lemma}\label{lemma:untangled}
    Let $U$ be a unitary operator in $\mathcal{M}_{n_1}(\C)\otimes \mathcal{M}_{n_2}(\C)$. Suppose there is a unitary operator $U_1\in \mathcal{M}_{n_1}(\C)$ such that 
    \begin{equation}
        U(A\otimes I_{n_2})U ^\dagger=U_1AU_1^\dagger\otimes I_{n_2}\label{eq:U1}
    \end{equation}
   for all $A\in \mathcal{M}_{n_1}(\C)$. Then there is unitary operator $U_2\in \mathcal{M}_{n_2}(\C)$ such that $U=U_1\otimes U_2$.
\end{lemma}

\begin{proof}
Let $W=(U_1^\dagger\otimes I_{n_2})U$. \Cref{eq:U1} implies that $A\otimes I_{n_2}=W^\dagger(A\otimes I_{n_2})W$ for all $A\in \mathcal{M}_{n_1}(\C)$. So $W\in \big(\mathcal{M}_{n_1}(\C)\otimes I_{n_2} \big)'=I_{n_1}\otimes \mathcal{M}_{n_2}(\C)$. Since $W$ is unitary, we conclude that $W=I_{n_1}\otimes U_2$ for some unitary $U_2\in \mathcal{M}_{n_2}(\C)$. It follows that $U=U_1\otimes U_2$.
\end{proof}

Recall from \cref{sec:preliminaries} that the set $\mathcal{B}_k$ contains all $k$-qutrit Pauli operators that generate the $k$-qutrit Pauli group $\mathcal{P}_k$, and that the set $\Clifford_n$ of $n$-qutrit Clifford circuits is generated by $H$, $S$, $\CZ$, and $-\omega$ via matrix multiplication and tensor product. 
For a composite system of $k_1+k_2$ qutrits, we identify $\mathcal{M}_{3^{k_1+k_2}}(\C)$ with $\mathcal{M}_{3^{k_1}}(\C)\otimes \mathcal{M}_{3^{k_2}}(\C)$. Under this identification, it is easy to see that $C\otimes I_{3^{k_2}}\in \mathcal{C}_{k_1+k_2}$ if and only if $C\in \mathcal{C}_{k_1}$.

\begin{proposition} \label{prop:untangled}
     Let $k_1, k_2\in \N$, and let $C\in \mathcal{C}_{k_1+k_2}$. Suppose that for every $B \in \mathcal{B}_{k_1}$ there exists a $B' \in \mathcal{P}_{k_1}$ such that
     \begin{equation}
         C(B \otimes I_{3^{k_2}})C^\dagger = B' \otimes I_{3^{k_2}}. \label{eq:CBC}
     \end{equation}
     Then there exist $C_1 \in \Clifford_{k_1}$ and $C_2 \in \Clifford_{k_2}$ such that $C = C_1\otimes C_2$, up to scalar.
\end{proposition}

\begin{proof}
We denote by $\tr_2$ the partial trace $id\otimes \tr$, which traces out the second system $\mathcal{M}_{3^{k_2}}(\C)$. Since $\mathcal{B}_{k_1}$ generates $\mathcal{P}_{k_1}$, it follows from \Cref{eq:CBC} that the mapping 
\begin{align*}
    \Phi: \mathcal{B}_{k_1}&\rightarrow \mathcal{M}_{3^{k_1}}(\C)\\
    B &\mapsto \frac{1}{3^{k_2}}\tr_{2}\big( C(B \otimes I_{3^{k_2}})C^\dagger \big)
\end{align*}
extends to an automorphism $\Phi$ of $\mathcal{P}_{k_1}$ such that
\begin{equation*}
     C(B \otimes I_{3^{k_2}})C^\dagger = \Phi(B) \otimes I_{3^{k_2}}.
\end{equation*}
Hence, there exists a Clifford operator $C_1 \in \Clifford_{k_1}$ such that $\Phi(B)=C_1 B C_1^\dagger$ for all $B\in \mathcal{P}_{k_1}$. Since $\mathcal{P}_{k_1}$ spans $\mathcal{M}_{3^{k_1}}(\C)$, we conclude that 
\begin{equation*}
    C(B \otimes I_{3^{k_2}})C^\dagger = C_1 B C_1^\dagger \otimes I_{3^{k_2}}
\end{equation*}
for all $B\in \mathcal{M}_{3^{k_1}}(\C)$. Then by \Cref{lemma:untangled}, $C=C_1\otimes C_2$ for some unitary $C_2\in \mathcal{M}_{3^{k_2}}(\C)$. Since $C$ and $C_1$ are both Clifford operators, it follows that $C_2\in \Clifford_{k_2}$.
\end{proof}

\begin{lemma}
The dirty gate resulting from pushing a CZ through two consecutive $B$ boxes is separable and only has Paulis on the first qutrit:
\begin{equation}
    \scalebox{1}{\input{figures/Generalizability/CZBBdir.tikz}}
    \label{eq:dir-CZBB}
\end{equation}
where $\dir = C_1\otimes C_2$ and $C_1$ is a Pauli.
\label{lem:CZ-BB-dir}
\end{lemma}

\begin{proof}
Based on \cref{lem:CZ.B.B-CZ.D.D-updated}, we can find the residual dirty gates after pushing a $\CZ$ through two consecutive $B$ boxes by left-appending $\left(I\otimes B_{a,b+2c}^{-1}\right)$ and $\left(B_{c,d+2a}^{-1}\otimes I\right)$ to both sides of~\eqref{eq:BB} as we see in~\eqref{eq:dir-CZBB}.

By \cref{prop:untangled} and the single-qutrit Clifford actions in \cref{def:Single-qutrit-action}, $\dir$ is separable and only has Paulis on the top qutrits when $\dir$ maps $X \otimes I_9$ to $\omega^{t_1}X \otimes I_9$ and $Z \otimes I_9$ to $\omega^{t_2}Z \otimes I_9$, for $t_1,t_2 \in \Z_3$.
  
First, let us find the preimage of $Z\otimes I_9$ in the lefthand side of \eqref{eq:dir-CZBB}. Note that

\begin{align}
    \left(B_{a,b+2c}^{-1}\right)^\dagger(X^aZ^{b+2c} \otimes Z)\left(B_{a,b+2c}^{-1}\right)&=\left(B_{a,b+2c}\right)(X^aZ^{b+2c} \otimes Z)\left(B_{a,b+2c}\right)^{\dagger}\xlongequal[]{\cref{fig:two-qutrit-normal-boxes-auto}}Z\otimes I.
    \label{eq:BB-dir-1}\\[2 em]
    \left(B_{c,d+2a}^{-1}\right)^\dagger(\omega^{2ac}X^cZ^{d+2a} \otimes Z)\left(B_{c,d+2a}^{-1}\right)&=\omega^{2ac}\left(B_{c,d+2a}\right)(X^cZ^{d+2a} \otimes Z)\left(B_{c,d+2a}\right)^{\dagger}\xlongequal[]{\cref{fig:two-qutrit-normal-boxes-auto}}\omega^{2ac}Z\otimes I.
    \label{eq:BB-dir-2}
\end{align}

\eqref{eq:BB-dir-1} and \eqref{eq:BB-dir-2} imply that (reading the updating of the Paulis from right-to-left):

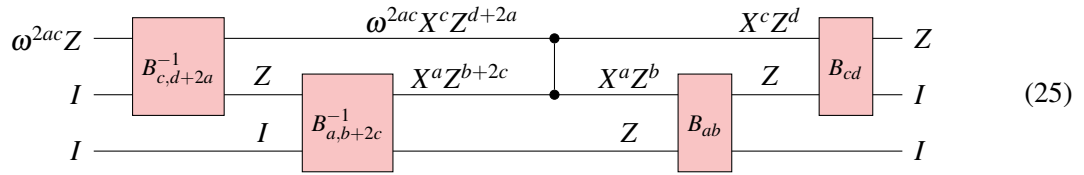
\begin{equation}
\scalebox{1}{\input{figures/NormalForm/CZ.B.B.tikz}}
    \label{eq:preimage-CZ.BB-ZII}
\end{equation}
    
Next, let us find the preimage of $X\otimes I_9$ in the lefthand side of \eqref{eq:dir-CZBB}. According to the third column of \cref{fig:two-qutrit-normal-boxes-auto}, we have

\[
    \scalebox{1}{\input{figures/Generalizability/CZBBdir2.tikz}}
\]

Since $2(b+2c)+a^2 - 1 = 2b + 4c + a^2 - 1 \equiv_3 2b+c+a^2-1$, we have
    
\begin{align}
&\left(B_{a,b+2c}^{-1}\right)^\dagger\left(\omega^{a(b+2c)}X^{2a}Z^{c+2b+a^2-1} \otimes X\right)\left(B_{a,b+2c}^{-1}\right)\notag\\
&=\left(B_{a,b+2c}\right)\left(\omega^{a(b+2c)}X^{2a}Z^{2(b+2c)+a^2-1} \otimes X\right)\left(B_{a,b+2c}\right)^{\dagger}\xlongequal[]{\cref{fig:two-qutrit-normal-boxes-auto}}X\otimes I.
    \label{eq:BB-dir-3}
\end{align}

Since $2(d+2a)+c^2 - 1 = 2d + 4a + c^2 - 1 \equiv_3 2d+a+c^2-1$, we have

    \begin{align}
&\left(B_{c,d+2a}^{-1}\right)^\dagger\left(\omega^{-2ac}\omega^{c(d+2a)}X^{2c}Z^{a+2d+c^2-1} \otimes X\right)\left(B_{c,d+2a}^{-1}\right)\notag\\
&=\omega^{-2ac}\left(B_{c,d+2a}\right)\left(\omega^{c(d+2a)}X^{2c}Z^{2(d+2a)+c^2 - 1} \otimes X\right)\left(B_{c,d+2a}\right)^{\dagger}\xlongequal[]{\cref{fig:two-qutrit-normal-boxes-auto}}\omega^{-2ac}X\otimes I.
    \label{eq:BB-dir-4}
\end{align}

\eqref{eq:BB-dir-3} and \eqref{eq:BB-dir-4} imply that

\begin{equation}
    \scalebox{1}{\input{figures/Generalizability/CZ.B.B2.tikz}}
    \label{eq:preimage-CZ.BB-XII}
\end{equation}

Putting \eqref{eq:dir-CZBB}, \eqref{eq:preimage-CZ.BB-ZII}, and \eqref{eq:preimage-CZ.BB-XII} together, we have

\begin{equation}
\scalebox{1}{\input{figures/Generalizability/CZBBdir3.tikz}}
    \label{eq:CZ.BB-dir}
\end{equation}
    \label{obs:B-boxes}
\end{proof}

\begin{lemma}
The dirty gate resulting from pushing a CZ through two consecutive $D$ boxes is separable and only has Paulis on the third qutrit:
\begin{equation}
    \scalebox{1}{\input{figures/Generalizability/CZDDdir.tikz}}
    \label{eq:dir-CZDD}
\end{equation}
where $\dir = C_1\otimes C_2$ and $C_2$ is a Pauli.
\label{lem:CZ-DD-dir}
\end{lemma}

\begin{proof}
    Based on \cref{lem:CZ.B.B-CZ.D.D-updated}, we can find the residual dirty gates after pushing a $\CZ$ through two consecutive $D$ boxes by left-appending $\left(D_{a,b+2c}^{-1}\otimes I\right)$ and  $\left(I\otimes D_{c,d+2a}^{-1}\right)$ to both sides of \eqref{eq:DD} as we see in \eqref{eq:dir-CZDD}.

According to the dirty gate restrictions in \cref{fig:dirty-gate-restrictions-1,fig:dirty-gate-restrictions-2}, our proposed $D$ box construction is suitable for Clifford normalization if there is no single $H$ gate acting on the bottom wire of $\dir$. One way to do this is to show that $\dir$ is a separable operator and the bottom wire is only acted on by Pauli $X$ or $Z$ gates. By \cref{prop:untangled} and the single-qutrit Clifford actions in \cref{def:Single-qutrit-action}, it suffices to show that $\dir$ maps $I_9 \otimes X$ to $I_9 \otimes \omega^{t_1}X$ and $I_9 \otimes Z$ to $I_9 \otimes \omega^{t_2}Z$, for $t_1,t_2 \in \Z_3$.

First, let us find the preimage of $I_9 \otimes Z$ in the lefthand side of \eqref{eq:dir-CZDD}. According to the third column of \cref{fig:two-qutrit-normal-boxes-auto}, we have

\begin{equation}
\scalebox{1}{\input{figures/Generalizability/CZDDdir2.tikz}}
    \label{eq:preimage-CZ.DD-IIZ}
\end{equation}

Next, let us find the preimage of $I_9 \otimes X$ in the lefthand side of \eqref{eq:dir-CZDD}. According to the second and third columns of \cref{fig:two-qutrit-normal-boxes-auto}, a $D_{ab}$ box maps $X \otimes X^aZ^b$ to $I \otimes X$, for all $a,b\in \Z_3$. It follows that

\begin{equation}
\scalebox{1}{\input{figures/Generalizability/CZDDdir3.tikz}}
    \label{eq:preimage-CZ.DD-IIX}
\end{equation}

Putting \eqref{eq:dir-CZDD}, \eqref{eq:preimage-CZ.DD-IIZ}, and \eqref{eq:preimage-CZ.DD-IIX} together, we have

\begin{equation}
\scalebox{1}{\input{figures/Generalizability/CZDDdir4.tikz}}
    \label{eq:CZ.DD-dir}
\end{equation}

\end{proof}

\begin{remark}
    For the proofs of \cref{lem:CZ-BB-dir,lem:CZ-DD-dir} to work, we need the partial automorphism of $B$ and $D$ boxes specified in the second and third columns of \Cref{fig:two-qutrit-normal-boxes-auto}.
\end{remark}

\begin{lemma}
    For $a, b, c, d \in \Z_3$ and $\dir' \in \Clifford_2$,
    \begin{align}
       \input{figures/Generalizability/CZ.B.BLHS.tikz}\quad&=\quad\input{figures/Generalizability/CZ.B.BRHS.tikz}\label{eq:CZ.B.B} \\[2 em]
       \input{figures/Generalizability/CZ.D.DLHS.tikz}\quad&=\quad\input{figures/Generalizability/CZ.D.DRHS.tikz}\label{eq:CZ.D.D}
    \end{align}
    \label{lem:CZ.B.B-lem:CZ.D.D}
\end{lemma}

\begin{proof}
    By \cref{lem:CZ.B.B-CZ.D.D-updated}, we can find the updated normal boxes in \eqref{eq:BB} and \eqref{eq:DD}. When pushing a $\CZ$ through two consecutive $B$ boxes, by \cref{lem:CZ-BB-dir,prop:untangled}, $\dir$ is a separable operator. According to \cref{app:generator-actions} and \eqref{eq:CZ.BB-dir}, $\dir = Z^{2ac}X^{2ac}\otimes \dir'$, with some $\dir' \in \Clifford_2$. When pushing a $\CZ$ through two consecutive $D$ boxes, by \cref{lem:CZ-DD-dir,prop:untangled}, $\dir$ is a separable operator. According to \cref{app:generator-actions} and \eqref{eq:CZ.DD-dir}, $\dir = \dir'\otimes Z^{ac}$, with some $\dir' \in \Clifford_2$.
\end{proof}

Based on the partial Pauli automorphisms specified in \cref{fig:two-qutrit-normal-boxes-auto}, we find concrete implementations of $B$ and $D$ boxes, as shown in \cref{fig:Z-Normal-Boxes,fig:X-Normal-Boxes}. Moreover, the dirty gates after pushing a $\CZ$ through two consecutive $B$ and $D$ boxes satisfy the restrictions outlined in \cref{fig:dirty-gate-restrictions-1,fig:dirty-gate-restrictions-2,fig:dirty-gate-restrictions-3}.

\paragraph{Summary}

For the whole procedure (i.e., Clifford normalization) to work, we need to derive all box relations following \cref{proc:find-d}, and check that the residual dirty gates satisfy the restrictions given in \cref{fig:dirty-gate-restrictions-1,fig:dirty-gate-restrictions-2,fig:dirty-gate-restrictions-3}. This can be done by inspecting every box relation listed in \cref{app:relations}. With these specific designs of normal boxes, we complete the task of designing the normal form for an arbitrary qutrit Clifford operator. As a result, the proof of completeness follows from all the box relations.
\end{document}